\documentclass[11pt]{article}
\usepackage[utf8]{inputenc}
\usepackage{fullpage}
\usepackage{graphicx}
\usepackage{amsthm}
\usepackage{amsmath}
\usepackage{amsfonts}
\usepackage{amssymb}
\usepackage[pagebackref,linktocpage]{hyperref}

\renewcommand{\backref}[1]{}

\renewcommand{\backrefalt}[4]{%
\ifcase #1 %
\or
[p.\ #2]%
\else
[pp.\ #2]%
\fi}

\usepackage{caption}
\usepackage{subcaption}
\usepackage{qcircuit}
\usepackage{tikz}
\usepackage{thm-restate}
\usepackage{dramatist}

\newtheorem{theorem}{Theorem}
\newtheorem{definition}[theorem]{Definition}

\newtheorem{proposition}[theorem]{Proposition}

\newcommand{\ket}[1]{\vert #1 \rangle}
\newcommand{\bin}[0]{\{0,1\}}

\title{A Qubit, a Coin, and an Advice String Walk Into a Relational Problem}
\author{
Scott Aaronson\thanks{University of Texas at Austin and OpenAI. \ Email:
\texttt{aaronson@cs.utexas.edu}. \ Supported by a Vannevar Bush Fellowship from the US
Department of Defense, the Berkeley NSF-QLCI CIQC Center, a Simons Investigator Award, and the Simons
\textquotedblleft It from Qubit\textquotedblright\ collaboration.}
\and
\stepcounter{footnote}Harry Buhrman\thanks{QuSoft, CWI, and University of Amsterdam. \ Email:
\texttt{buhrman@cwi.nl}. \ Supported in part by the Dutch Research Council (NWO/OCW), Gravitation Programmes Quantum Software Consortium  (project number 024.003.037) and Networks (project number 024.002.003).}\and
William Kretschmer\thanks{Simons Institute for the Theory of Computing and University of California, Berkeley. \ Email:
\texttt{kretsch@berkeley.edu}. \ Supported by an NDSEG Fellowship and a Simons Quantum Postdoctoral Fellowship. Most of this work done while a student at UT Austin.}
}
\date{}

\begin{document}

\maketitle

\begin{abstract}
Relational problems (those with many possible valid outputs) are different from decision problems, but it is easy to forget just \textit{how} different. \ This paper initiates the study of $\mathsf{FBQP/qpoly}$, the class of relational problems solvable in quantum polynomial-time with the help of polynomial-sized quantum advice, along with its analogues for deterministic and randomized computation ($\mathsf{FP}$, $\mathsf{FBPP}$) and advice ($\mathsf{/poly}$, $\mathsf{/rpoly}$).

Our first result is that $\mathsf{FBQP/qpoly} \ne \mathsf{FBQP/poly}$, \textit{unconditionally}, with no oracle---a striking contrast with what we know about the analogous decision classes. \ The proof repurposes the separation between quantum and classical one-way communication complexities due to Bar-Yossef, Jayram, and Kerenidis. \ We discuss how this separation raises the prospect of near-term experiments to demonstrate ``quantum information supremacy,'' a form of quantum supremacy that would not depend on unproved complexity assumptions.

Our second result is that $\mathsf{FBPP} \not\subset \mathsf{FP/poly}$---that is, \textit{Adleman's Theorem fails for relational problems}---unless $\mathsf{PSPACE} \subset \mathsf{NP/poly}$. \ Our proof uses $\mathsf{IP}=\mathsf{PSPACE}$ and time-bounded Kolmogorov complexity. \ On the other hand, we show that proving $\mathsf{FBPP} \not\subset \mathsf{FP/poly}$ will be hard, as it implies a superpolynomial circuit lower bound for $\mathsf{PromiseBPEXP}$.

We prove the following further results:
\begin{itemize}
    \item \textit{Unconditionally}, $\mathsf{FP} \ne \mathsf{FBPP}$ and $\mathsf{FP/poly} \ne \mathsf{FBPP/poly}$ (even when these classes are carefully defined).
    \item $\mathsf{FBPP/poly} = \mathsf{FBPP/rpoly}$ (and likewise for $\mathsf{FBQP}$). \ For \textit{sampling} problems, by contrast, $\mathsf{SampBPP/poly} \ne \mathsf{SampBPP/rpoly}$ (and likewise for $\mathsf{SampBQP}$).
\end{itemize}
\end{abstract}

\newpage
\section{Introduction}

Here is a basic and underappreciated fact: there are computational problems---not distributed or cryptographic tasks, but just pure computational problems---that \textit{provably} admit only randomized solutions. \ One simple example is: ``output an $n$-bit Kolmogorov-random string.'' \ Another example is: ``given as input a halting Turing machine $M$, output any string \emph{other than} what $M$ outputs when run on its own description.''

Both of these are \textit{relational problems}, defined by a relation $R\subseteq \bin^* \times \bin^*$. \ Given an input $x$, the goal in such a problem is to output any $y$ such that $(x,y)\in R$. \ The class $\mathsf{FP}$ consists of all relations $R$ for which there exists a deterministic polynomial-time algorithm to find a $y$ such that $(x,y)\in R$ whenever one exists. \ (Ironically, the $\mathsf{F}$ stands for ``functional,'' even though the whole point with relational problems is that they need \emph{not} be functions.)

It is trickier to define $\mathsf{FBPP}$ and $\mathsf{FBQP}$, the relational analogues of $\mathsf{BPP}$ and $\mathsf{BQP}$ respectively. \ For unlike with decision problems, we can no longer amplify success probabilities by taking majorities, so different allowed error probabilities could lead to different complexity classes. \ For this reason, a wide variety of definitions of $\mathsf{FBPP}$ have appeared in the literature \cite{goldreich:pbpp,aar:samp,Man19-thesis,my23-adv-owfs,ilw23}. \ Having said that, there is one choice that seems more natural than others, which Aaronson \cite{aar:samp} made more than a decade ago and which we follow here.\footnote{We also found that GPT-4 \cite{gpt4}, when prompted to give a definition for $\mathsf{FBPP}$, settled on one similar to Definition \ref{FBPPdef}. See Appendix \ref{GPT} for the transcript.}

Call the relation $R\subseteq \bin^* \times \bin^*$ \textit{polynomially-bounded} if there exists a polynomial $p$ such that $|y| \le p(|x|)$ for all $(x,y)\in R$. \ Then:

\begin{definition}
\label{FBPPdef}
$\mathsf{FBPP}$ is the class of polynomially-bounded relations $R\subseteq \bin^* \times \bin^*$ for which there exists a polynomial-time randomized algorithm $A$ such that for all $x$ for which there exists a $y$ with $(x,y)\in R$ and all $\varepsilon > 0$,
$$ \Pr[(x,A(x,0^{1/\varepsilon}))\in R]\ge 1-\varepsilon,$$
where the probability is over $A$'s outputs. \ $\mathsf{FBQP}$ is exactly the same except that $A$ can now be a quantum algorithm.
\end{definition}

A few comments on this definition: we require $A$ to succeed for any given $\varepsilon>0$ in order to avoid problems being in $\mathsf{FBPP}$ or $\mathsf{FBQP}$ for ``accidental'' reasons, i.e.\ that the fraction of strings $y$ such that $(x,y)\not\in R$ happens to fall below some arbitrary threshold. \ We allow time polynomial in $1/\varepsilon$ because, as we'll see, there are natural reductions that need such time. \ We demand that $R$ be polynomially-bounded because otherwise, $A$ might achieve smaller and smaller error probabilities $\varepsilon$ by outputting longer and longer strings, rather than ``doing better and better on the same strings,'' which is not what we intuitively wanted when we allowed $\operatorname{poly}(n,1/\varepsilon)$ time. \ Finally, we do not require that membership in the relation be efficiently verifiable, in contrast to Goldreich's definition \cite[Definition 3.1]{goldreich:pbpp}. \ This is for fairness to quantum algorithms: we want $\mathsf{FBQP}$ to contain the relational analogues of problems like BosonSampling \cite{aark} and Random Circuit Sampling \cite{achen} that have played a central role in recently claimed demonstrations of quantum computational supremacy \cite{arute,ustc}.
\ However, it seems unlikely that such problems can admit efficient verification of membership in the relation.\footnote{For example, if the relation $R \in \mathsf{FBQP}$ defined in \cite[Corollary 5.10]{aark} had efficient verification of membership, then the Gaussian Permanent Estimation problem $|\mathrm{GPE}|^2_\pm$ \cite[Problem 1.2]{aark} would be solvable in $\mathsf{FP}^\mathsf{PH}$, thus refuting either the Permanent-of-Gaussians Conjecture \cite[Conjecture 1.5]{aark}, the Permanent Anti-Concentration Conjecture \cite[Conjecture 1.6]{aark}, or $\mathsf{P^{\# P}} \not \subset \mathsf{PH}$.}

Already with $\mathsf{FBPP}$ and $\mathsf{FBQP}$, some interesting phenomena rear their heads: for example, we'll observe in Section \ref{FBPPSEC} that $\mathsf{FP}\ne \mathsf{FBPP}$, unconditionally. \ Note that, because of the requirement to succeed with probability $1-\varepsilon$ for any $\varepsilon > 0$, this does \textit{not} immediately follow from the examples with which we opened the paper, but it does follow from modifications of those examples, involving time-bounded Kolmogorov complexity or the time-bounded halting problem.

The message of this paper is that the story of $\mathsf{FBPP}$ and $\mathsf{FBQP}$ becomes wilder still---even more divergent from expectations formed from decision problems---once we bring classical and quantum \textit{advice} into the picture.

\subsection{Advice Classes}

Karp and Lipton \cite{kl} introduced the nonuniform complexity class $\mathsf{P/poly}$ and proved the famous theorem that $\mathsf{NP}\subset \mathsf{P/poly}$ would imply the collapse of the polynomial hierarchy. \ Meanwhile, Adleman \cite{adleman:bpp} proved that $\mathsf{BPP}\subset \mathsf{P/poly}$. \ Indeed, it is not hard to see that
$$ \mathsf{BPP/rpoly} = \mathsf{BPP/poly} = \mathsf{P/rpoly} = \mathsf{P/poly},$$
where $\mathsf{/rpoly}$ means ``with polynomial-sized randomized advice,'' and $\mathsf{P/rpoly}$ is the class of languages that admit nonuniform polynomial-time bounded-error randomized algorithms in which the only randomness comes from the advice. Note that the  $\mathsf{/rpoly}$ advice is at least as powerful as the $\mathsf{/poly}$ advice as the $\mathsf{/poly}$ advice can be seen as a distribution with probability concentrated on a single string.

When we come to $\mathsf{BQP}$, it's natural to ask what happens when the advice can be a quantum state on polynomially many qubits---perhaps a highly entangled state that's intractable to prepare on one's own. \ To capture this question, in 2003 Nishimura and Yamakami \cite{ny} defined the class $\mathsf{BQP/qpoly}$, or Bounded-Error Quantum Polynomial-Time with polynomial-size quantum advice.

\begin{definition}
$\mathsf{BQP/qpoly}$ is the class of all languages $L\subseteq \bin^*$ for which there exists a polynomial-time quantum algorithm $A$, a polynomial $p$, and an infinite list of advice states $\{\ket{\psi_n}\}_{n\ge 1}$, where $\ket{\psi_n}$ is on $p(n)$ qubits, such that for all $n$ and all $x\in \bin^n$,
$$ \Pr[A(x,\ket{\psi_n})=L(x) ]\ge \frac{2}{3}.$$
\end{definition}

Studying $\mathsf{BQP/qpoly}$ is one way to formalize the old question of ``how much information is in an $n$-qubit state.'' \ On the one hand, if we think of an $n$-qubit state $\ket{\psi}$ as a unit vector in $\mathbb{C}^{2^n}$, then it seems $\ket{\psi}$ could provide an exponential amount of information---say, about every possible input $x\in\bin^n$. \ On the other hand, Holevo's Theorem \cite{holevo} implies that we can encode at most $n$ bits into $n$ qubits, in such a way that they can be reliably retrieved later by measuring them.

So then, does $\mathsf{BQP/qpoly}$ collapse with $\mathsf{BQP/poly}$---that is, $\mathsf{BQP}$ with polynomial-sized \textit{classical} advice---or could it be vastly more powerful?

A priori, it's not even obvious that $\mathsf{BQP/qpoly} \ne \mathsf{ALL}$, where $\mathsf{ALL}$ is the class of all languages. \ Underscoring this worry, it's easy to show (for example) that $\mathsf{PostBQP/qpoly}=\mathsf{ALL}$, where $\mathsf{PostBQP}$ means quantum polynomial time with postselected measurements. \ To see this, given a language $L$ and an input length $n$, we just need to consider the advice state
$$ \ket{\psi_n} = \frac{1}{\sqrt{2^n}} \sum_{z\in\bin^n} \ket{z} \ket{L(z)}, $$
where $L(z)=1$ if $z\in L$ and $L(z)=0$ otherwise. \ Then given an input $x\in\bin^n$, we first measure $\ket{\psi_n}$ in the standard basis, then postselect on getting the outcome $\ket{x} \ket{L(x)}$.




Despite the sometimes unsettling power of randomized and quantum advice, in 2004, Aaronson \cite{aar:adv} proved that $\mathsf{BQP/qpoly}\subseteq\mathsf{PostBQP/poly}$. \ Since Aaronson \cite{aar:pp} also showed that $\mathsf{PostBQP} = \mathsf{PP}$, and since adding deterministic advice ``commutes'' with standard complexity class inclusions, this can be stated equivalently as $\mathsf{BQP/qpoly}\subseteq\mathsf{PP/poly}$.

This upper bound on the power of $\mathsf{BQP/qpoly}$ has a few implications. \ First, it immediately implies that $\mathsf{BQP/qpoly} \ne \mathsf{ALL}$, since $\mathsf{PP/poly}\ne \mathsf{ALL}$ is easy to show by a counting argument. \ Second, it means that there is no hope, in the present state of complexity theory, of proving that $\mathsf{BQP/poly} \ne \mathsf{BQP/qpoly}$. \ For any such proof would imply $\mathsf{BQP/poly} \ne \mathsf{PP/poly}$, and hence (for example) that $\mathsf{PSPACE}$ does not have polynomial-size circuits. \ At best, one could hope to show that $\mathsf{BQP}^A\mathsf{/poly} \ne \mathsf{BQP}^A\mathsf{/qpoly}$ for some oracle $A$. \ As it happens, even \textit{this} is still open, although Aaronson and Kuperberg \cite{ak} showed the existence of a unitary oracle $U$ such that $\mathsf{BQP}^U\mathsf{/poly} \ne \mathsf{BQP}^U\mathsf{/qpoly}$, and there's been recent progress toward replacing this with an ordinary classical oracle \cite{feffermankimmel,natarajannirkhe}.

\subsection{Relational Complexity Classes with Advice}

In quantum computing, it has repeatedly been found that it's easier to see the advantages of quantum algorithms over classical ones once we switch attention from decision problems to relational and sampling problems. \ This is what happened, for example, with BosonSampling \cite{aark}, Random Circuit Sampling \cite{achen}, and other sampling-based approaches to demonstrating quantum supremacy. \ It is also what happened with the recent breakthrough of Yamakawa and Zhandry \cite{yz}, which achieved an exponential quantum speedup relative to a random oracle---but only by switching from decision problems (where the \textit{Aaronson-Ambainis Conjecture} \cite{aa:struc} asserts that no such separation is possible) to $\mathsf{NP}$ search problems, a particular kind of relational problem.\footnote{If one just wants a superpolynomial quantum speedup relative to a random oracle for \textit{some} relational problem---not necessarily an $\mathsf{NP}$ search problem---then Aaronson \cite{aar:ph} showed that the problem of outputting large Fourier coefficients of a random Boolean function does the job.}

In this paper, then, we do something that could've been done at any point in the past 20 years, but apparently wasn't: namely, we ask about the advantages of quantum over classical \textit{advice} on relational problems.

\begin{definition}
\label{fbqpqpolydef}
$\mathsf{FBQP/qpoly}$ is the class of polynomially-bounded relations $R\subseteq \bin^* \times \bin^*$ for which there exists a polynomial-time quantum algorithm $Q$, a polynomial $p(n,m)$, and an infinite list of advice states $\{\ket{\psi_{n,m}}\}_{n,m\ge 1}$, where $\ket{\psi_{n,m}}$ is on $p(n,m)$ qubits, such that for all $x$ for which there exists a $y$ such that $(x,y)\in R$ and all $m$,
$$ \Pr[(x,Q(x,0^m,\ket{\psi_{n,m}}))\in R]\ge 1-\frac{1}{m}.$$
\end{definition}

The one subtlety in this definition is that the advice state $\ket{\psi_{n,m}}$ is allowed to depend, not only on the input length $n$, but on the desired error probability $\varepsilon = 1/m$. \ We claim that this is simply the ``right'' choice: efficiency in this setting means time polynomial in $n$ and $1/\varepsilon$, so the advice ought to be allowed to depend on both parameters as well.

Of course, it doesn't make much sense to feed quantum advice to a classical complexity class (e.g., $\mathsf{FP/qpoly}$). \ On the other hand, it's sensible to consider $\mathsf{FP/rpoly}$: this corresponds to classical algorithms that \textit{only} get to use random bits if they come from the advice.

\begin{definition}
\label{fprpolydef}
$\mathsf{FP/rpoly}$ is the class of polynomially-bounded relations $R\subseteq \bin^* \times \bin^*$ for which there exists a polynomial-time deterministic classical algorithm $A$, a polynomial $p(n,m)$, and an infinite list of advice distributions $\{\mathcal{D}_{n,m}\}_{n,m\ge 1}$, where $\mathcal{D}_{n,m}$ is supported on $\bin^{p(n,m)}$, such that for all $x$ for which there exists a $y$ such that $(x,y)\in R$ and all $m$,
$$ \Pr_{r \sim \mathcal{D}_{n,m}}[(x,A(x,0^m,r))\in R]\ge 1-\frac{1}{m}.$$
\end{definition}

We can similarly define $\mathsf{FBPP/rpoly}$, $\mathsf{FBQP/poly}$, and other possible combinations; we omit the details.

\subsection{Our Results}

We show that switching attention to relational problems dramatically changes the picture of randomized and quantum computation in the presence of advice.

Our first result is that quantum advice \textit{unconditionally} provides more power than classical advice to solve relational problems:

\begin{restatable}{theorem}{qpolythmrestate}
\label{qpolythm}
$\mathsf{FBQP/qpoly} \ne \mathsf{FBQP/rpoly}$.
\end{restatable}

So in particular, $\mathsf{FBQP/qpoly} \ne \mathsf{FBQP/poly}$. \ Indeed, we shall see that $\mathsf{FBQP/qpoly}$ is not contained in $\mathsf{F}\mathcal{C}\mathsf{/poly}$ for arbitrarily powerful uniform complexity classes $\mathcal C$: for example, the class of all computable problems. \ This is despite the fact that $\mathsf{FBQP/qpoly}$ does \textit{not} equal $\mathsf{ALL}$---as can be seen, for example, by considering its restriction to Boolean-valued problems, where it coincides with $\mathsf{PromiseBQP/qpoly} \subseteq \mathsf{PromisePP/poly}$.

As we discuss in Section \ref{SUPREMACY}, this complexity class separation suggests the possibility of a near-term experiment, which would run an $\mathsf{FBQP/qpoly}$ protocol in order to check explicitly whether an entangled state of $n$ qubits (where, say, $n\approx 20$) encodes $\gg n$ bits of classical information. \ We hope further work will clarify whether such an experiment is feasible with current devices.

Theorem \ref{qpolythm} is nonconstructive, and does not give an explicit example of a relation in $\mathsf{FBQP/qpoly}$ but not $\mathsf{FBQP/rpoly}$ (not counting, e.g., the use of brute force to find the lexicographically first relation that works). \ We leave the ``explicitization'' of this separation as one of our central challenges.

Our second result shows that Adleman's Theorem \cite{adleman:bpp}, that $\mathsf{BPP} \subset \mathsf{P/poly}$, almost certainly does \textit{not} extend to relational problems:\footnote{After this manuscript first appeared, Ilango, Li, and Williams \cite{ilw23} implicitly established a conceptually similar result that $\mathsf{FBPP}\not\subset \mathsf{FP/poly}$ under plausible assumptions. They show a conditional lower bound for the range avoidance problem $\textsc{Avoid}$, which lies in $\mathsf{FBPP}$ (whenever the stretch is at least linear). Roughly, \cite[Theorem 28]{ilw23} shows that if subexponentially-secure indistinguishability obfuscation exits and $\mathsf{coNP}$ is not in $\mathsf{NP/poly}$ infinitely often, then $\textsc{Avoid} \not\in \mathsf{FP/poly}$. Comparatively, our result seems to weaken the assumption required to separate $\mathsf{FBPP}$ from $\mathsf{FP/poly}$, but only in the sense that non-collapse of $\mathsf{PH}$ is a better-tested assumption than the existence of indistinguishability obfuscation.}

\begin{restatable}{theorem}{fppolythmrestate}
If $\mathsf{FBPP}\subset \mathsf{FP/poly}$, then $\mathsf{PSPACE}\subset \mathsf{NP/poly}$ (and hence $\mathsf{PH}$ collapses).
\label{fppolythm}
\end{restatable}

We complement Theorem \ref{fppolythm} with a result showing that an unconditional proof of $\mathsf{FBPP} \not\subset \mathsf{FP/poly}$ is unlikely in the current state of complexity theory, as it would imply breakthrough circuit lower bounds:

\begin{restatable}{theorem}{fbppupperboundrestate}
\label{fbppupperboundthm}
$\mathsf{FBPP} \subseteq \mathsf{FP^{PromiseBPEXP}}$. Hence, if $\mathsf{PromiseBPEXP} \subset \mathsf{PromiseP/poly}$, then $\mathsf{FBPP} \subset \mathsf{FP/poly}$.
\end{restatable}

We also show that, when $\mathsf{FP}$ and $\mathsf{FBPP}$ are either both given advice or both \textit{not} given advice, the separation between them becomes unconditional:

\begin{restatable}{theorem}{fpneqfbpprestate}
$\mathsf{FP} \ne \mathsf{FBPP}$.
\label{fpneqfbpp}
\end{restatable}

\begin{restatable}{theorem}{fpneqfbpppolyrestate}
$\mathsf{FP/poly} \ne \mathsf{FBPP/poly}$.
\label{fpneqfbpppoly}
\end{restatable}

This underscores yet another difference between decision and relational problems: if $\mathcal{C}_1$ and $\mathcal{C}_2$ are two uniform classes of promise problems, then the question of whether $\mathcal{C}_1 \subset \mathcal{C}_2\mathsf{/poly}$ is equivalent to the question of whether $\mathcal{C}_1\mathsf{/poly} \subseteq \mathcal{C}_2\mathsf{/poly}$, since an advice string can just be appended to the input. \ With relational complexity classes such as $\mathsf{FBPP}$, however, this equivalence is no longer immediate, since it doesn't account for how the length of the advice can depend on the error bound $\varepsilon$.

A last question is whether our separation between classical and quantum advice, in the relational setting, extends to a separation between deterministic and randomized advice. \ We show that the answer is no:

\begin{restatable}{theorem}{rpolypolyequalrestate}
$\mathsf{FBPP/rpoly}=\mathsf{FBPP/poly} = \mathsf{FP/rpoly}$ and $\mathsf{FBQP/rpoly}=\mathsf{FBQP/poly}$.
\label{rpolypolyequal}
\end{restatable}

See Figure \ref{classfig} for the complexity class inclusion diagram that emerges from our results about relational classes.

\begin{figure}
\centering
\begin{tikzpicture}[x=5em,y=5em]
    \node at (0,3) {$\mathsf{FBQP/qpoly}$};
    \draw[thick,->] (0,2.25) -- (0,2.75);
    \node at (0.8,2.5) {Theorem \ref{qpolythm}};
    \node at (0,2) {$\mathsf{FBQP/rpoly} = \mathsf{FBQP/poly}$};
    \node at (0,1.75) {Theorem \ref{rpolypolyequal}};
    \draw[thick,dashed,->] (0,1) -- (0,1.5);
    \node at (0,0.75) {$\mathsf{FBPP/rpoly}=\mathsf{FBPP/poly} = \mathsf{FP/rpoly}$ };
    \node at (0,0.5) {Theorem \ref{rpolypolyequal}};
    \draw[thick,->] (0,-0.25) -- (0,0.25);
    \node at (0.8,0) {Theorem \ref{fpneqfbpppoly}};
    \node at (0,-0.5) {$\mathsf{FP/poly}$ };

    \node at (-3,0.75) {$\mathsf{FP^{PromiseBPEXP}}$};
    \draw[thick,dashed,->] (-2.1,-0.25) -- (-3,0.6);
    \node at (-3.2,0) {Theorem \ref{fbppupperboundthm}};
    \node at (-2,-0.5) {$\mathsf{FBPP}$};
    \draw[thick,->] (-2,-1.25) -- (-2,-0.75);
    \node at (-2.8,-1) {Theorem \ref{fpneqfbpp}};
    \node at (-2,-1.5) {$\mathsf{FP}$};
    \draw[thick,->] (-1.9,-0.25) -- (-0.7,0.5);
    \draw[thick,->] (-1.7,-1.4) -- (-0.5,-0.65);
    \draw[dashed,thick,->] (-1.6,-0.5) -- (-0.5,-0.5);
    \node at (-1,-0.25) {Theorem \ref{fppolythm}};
    \draw[thick,red] (-1.1,-0.6) -- (-0.9,-0.4);
    \draw[thick,red] (-0.9,-0.6) -- (-1.1,-0.4);
\end{tikzpicture}
    \caption{\label{classfig} Relationships among classes of relational problems considered in this paper. \ A solid arrow from $\mathcal{C}_1$ to $\mathcal{C}_2$ indicates strict containment $(\mathcal{C}_1 \subsetneq \mathcal{C}_2)$. \ A dashed arrow indicates a containment $\mathcal{C}_1 \subseteq \mathcal{C}_2$
    that we conjecture to be strict, but a proof of strictness would require a breakthrough in complexity theory. \ A crossed dashed arrow indicates non-containment $(\mathcal{C}_1 \not\subset \mathcal{C}_2)$ under the assumption that $\mathsf{PH}$ does not collapse.}
\end{figure}

We remark that several of our results are sensitive to the choices we made in defining $\mathsf{FBPP}$ and its variants, especially in regards to error reduction. \ In Section \ref{SMALLERROR}, we explore the consequences of choosing some alternative error bounds in Definition \ref{FBPPdef}. \ There, we find that Theorems \ref{fppolythm}, \ref{fpneqfbpp}, and \ref{fpneqfbpppoly}
no longer hold unconditionally if we demand exponential error reduction, meaning that the algorithm outputs a sample consistent with the relation with probability $1 - \varepsilon$ in time \textit{polylogarithmic}  in $1/\varepsilon$.\footnote{By contrast, Theorem \ref{qpolythm} is unaffected by such a change in definition, because the $\mathsf{FBQP/qpoly}$ algorithm used in our proof will turn out to be errorless.} \ So, our results could be interpreted in two different ways: either as showing a striking contrast between relational and decisional classes, or as showing the remarkable power of $\mathsf{FBPP}$ when we don't demand exponential error reduction. \ We leave it to the reader to decide, and hope that this work inspires more discussion about subtleties in the definitions.

\subsection{Quantum Communication Complexity}

As it turns out, essentially everything we need to prove Theorem \ref{qpolythm} was proved $20$ years ago, by Bar-Yossef, Jayram, and Kerenidis \cite{bjk}---though the fact that this is so is buried in their paper. \ These authors considered separations between randomized and quantum \textit{one-way communication complexities}. \ That is, they considered the setting where Alice has an input $x$, Bob has an input $y$, and Alice can send a message $m_x$ to Bob, which should then allow Bob to compute some joint property of $x$ and $y$.

Let $T$ be a task, which might be the evaluation of a Boolean function $f(x,y)$, but might also be a sampling or relational problem. \ We define $D^1(T)$, $R^1(T)$, and $Q^1(T)$ to be the minimum number of bits sent from Alice to Bob in any deterministic, bounded-error randomized, or bounded-error quantum one-way communication protocol respectively that lets Bob perform the task for all valid input pairs $(x,y)$ (with the number of bits maximized over all such input pairs). \ We assume no shared randomness or entanglement.

Clearly $D^1(T) \ge R^1(T) \ge Q^1(T)$ for all tasks $T$. \ A natural question is how large the separations between the measures can be. \ It's well-known that $D^1$ and $R^1$ can be exponentially separated: for example, for the $N$-bit EQUALITY function $\operatorname{EQ}$, we have $D^1(\operatorname{EQ})=N$ while $R^1(\operatorname{EQ})=O(\log N)$. \ But what about $R^1$ versus $Q^1$?

To study this, Bar-Yossef, Jayram, and Kerenidis \cite{bjk} defined a relation problem called \textit{Hidden Matching} or $\operatorname{HM}$. \ Here Alice is given a string $x\in\bin^N$ (with $N$ even), while Bob is given a perfect matching $y$ on the set $[N]$, consisting of $N/2$ edges. \ Bob's goal is to output $(i,j,x_i\oplus x_j)$ for \textit{some} edge $(i,j)\in y$. \ The key result is then the following:

\begin{theorem}[\cite{bjk}]
\label{bjkthm}
$Q^1(\operatorname{HM})=O(\log N)$, whereas $R^1(\operatorname{HM})=\Omega(\sqrt N)$.
\end{theorem}


Crucially for us, Bar-Yossef, Jayram, and Kerenidis actually proved the following stronger statement:

\begin{theorem}[{\cite[Proof of Theorem 4.1, page 373]{bjk}}]
\label{bjkthm_distributional}
Let $\mathcal M$ be any set of perfect matchings on $[N]$ that is pairwise edge-disjoint and satisfies $|\mathcal{M}| = \Omega(N)$. \ Let $\mu$ be the distribution over inputs to $\operatorname{HM}$ in which Alice's input is uniform in $\{0,1\}^N$ and Bob's input is uniform in $\mathcal{M}$. \ Then, any deterministic one-way protocol for $\operatorname{HM}$ that errs with probability at most $1/8$ with respect to $\mu$ requires $\Omega(\sqrt{N})$ bits of communication.
\end{theorem}

To prove Theorem \ref{qpolythm}, in Section \ref{MAIN}, we adapt Theorem \ref{bjkthm} to the setting of $\mathsf{FBQP/qpoly}$, treating the advice as one-way communication from an advisor to the $\mathsf{FBQP}$ algorithm.

To understand the situation more deeply, recall the result of Aaronson \cite{aar:adv} from before, that $\mathsf{BQP/qpoly}\subseteq\mathsf{PP/poly}$. \ A direct analogue of that result for one-way communication complexity \cite{aar:adv} says that $D^1(f)$ and $Q^1(f)$ are close whenever Bob's input is small:

\begin{theorem}[\cite{aar:adv}]
\label{aarccthm}
For all Boolean functions $f:\bin^n \times \bin^m \rightarrow \bin$ (partial or total),
$$D^1(f) = O(m Q^1(f) \log Q^1(f)).$$
\end{theorem}

This paper is pointing out that Theorem \ref{aarccthm}, and $\mathsf{BQP/qpoly}\subseteq\mathsf{PP/poly}$, both
 \textit{fail catastrophically} for sampling and relational problems. \ This seems not to have been known even to experts who we asked. \ One reason, perhaps, is that the original separation of Bar-Yossef, Jayram, and Kerenidis \cite{bjk} was partly overshadowed by the later work of Gavinsky et al.\ \cite{gkkrw}. \ The latter modified the Hidden Matching relational problem to obtain a partial Boolean function, called \textit{Boolean Hidden Matching} or $\operatorname{BHM}$. \ They then showed that $Q^1(\operatorname{BHM})=O(\log N)$ whereas $R^1(\operatorname{BHM})=\Omega(\sqrt{N})$.

We are calling attention to a surprising difference between the original Hidden Matching separation and the later Boolean Hidden Matching one. \ Namely: we can make Bob's input ``small'' (say, $O(\log n)$ bits) in the $\operatorname{HM}$ separation, \textit{even though we cannot do the same in the $\operatorname{BHM}$ separation}. \ For Boolean $f$, Theorem \ref{aarccthm} shows that an exponential gap between $D^1(f)$ and $Q^1(f)$ is possible only when Bob's input is ``large.''

\subsection{Other Proofs}

Let us make a few remarks about our other results, proved in Sections \ref{MAIN} and \ref{FBPPSEC}. \ To show that $\mathsf{FBPP/rpoly}=\mathsf{FBPP/poly}$, we just take deterministic advice that consists of $O(n/\varepsilon^2)$ independent samples from the randomized advice distribution, and then appeal to a Hoeffding and union bound. \ To show that $\mathsf{FP}\ne \mathsf{FBPP}$, we consider the problem of outputting an $n$-bit string with large time-bounded Kolmogorov complexity.\footnote{An alternative approach (not shown here) is to prove $\mathsf{FP} \neq \mathsf{FBPP}$ using a direct diagonalization. The core idea of this argument is captured in \cite[Section 3.1]{goldreich:pbpp}.} \ To show that $\mathsf{FBPP} \subseteq \mathsf{FP^{PromiseBPEXP}}$, we give a simple polynomial-time algorithm that builds a string in the relation one bit at a time, using the $\mathsf{PromiseBPEXP}$ oracle.

Finally, and most interestingly, to show that a ``relational Adleman's Theorem'' ($\mathsf{FBPP} \subset \mathsf{FP/poly}$) is unlikely to hold, we build on an old idea due to Buhrman and Torenvliet \cite{buhrtor}. \ We show that, if the problem of generating strings of high conditional time-bounded Kolmogorov complexity were in $\mathsf{FP/poly}$, then in the $\mathsf{IP}=\mathsf{PSPACE}$ protocol \cite{shamir}, we could replace the randomized verifier by a deterministic polynomial-size circuit. \ Roughly speaking, the verifier replaces each random challenge with a string of high time-bounded Kolmogorov complexity conditioned on the prior transcript of the protocol. \ To argue that this derandomization is sound, we just have to show that the ``bad'' choices of randomness (i.e. those that cause the verifier to accept when it should reject) all have low conditional time-bounded Kolmogorov complexity. \ We complete the proof by observing that this derandomization would put $\mathsf{PSPACE}$ into $\mathsf{NP/poly}$.

\subsection{Sampling Problems\label{SAMPLINTRO}}

We conclude with some results about sampling problems, which are closely related to relation problems. \ 
A sampling problem is defined by a collection of probability distributions $\mathcal D_x$. \ Given an input $x$, the goal is to output a sample from $\mathcal D_x$, either exactly or approximately. \ Like for relational problems, we call a sampling problem $S=\{ \mathcal D_x \}_{x\in \bin^*}$ \textit{polynomially-bounded} if there exists a polynomial $p$ such that for every $x$, $\mathcal{D}_x$ is a distribution over strings of length at most $p(|x|)$. \ Again following Aaronson \cite{aar:samp}, we define the basic complexity class like so:

\begin{definition}
$\mathsf{SampBQP}$ is the class of polynomially-bounded sampling problems $S=\{ \mathcal D_x \}_{x\in \bin^*}$ for which there exists a polynomial-time quantum algorithm $Q$ such that for all $x$ and all $\varepsilon > 0$,
$$ \| \mathcal D_Q(x,0^{1/\varepsilon}) - \mathcal D_x \| \le \varepsilon,$$
where $\mathcal D_Q(x,0^{1/\varepsilon})$ represents $Q$'s output distribution on input $(x,0^{1/\varepsilon})$ and $\| \|$ represents total variation distance.
\end{definition}

Again, we can consider the classical analogue $\mathsf{SampBPP}$ (the deterministic version, $\mathsf{SampP}$, doesn't make much sense). \ We can also combine with deterministic, randomized, and quantum advice like in Definition \ref{fbqpqpolydef}, to get $\mathsf{SampBPP/poly}$, $\mathsf{SampBQP/qpoly}$, and so on. \ For example:

\begin{definition}
$\mathsf{SampBPP/rpoly}$ is the class of polynomially-bounded sampling problems $S=\{ \mathcal D_x \}_{x\in \bin^*}$ for which there exists a polynomial-time randomized algorithm $A$, a polynomial $p(n,m)$, and an infinite list of advice distributions $\{ \mathcal{D}_{n,m}\}_{n,m \ge 1}$, where $\mathcal{D}_{n,m}$ is supported on $\{0,1\}^{p(n,m)}$, such that for all $x$ and all $m$,
$$ \| \mathcal D_A(x,0^m,\mathcal{D}_{n,m}) - \mathcal D_x \| \le \frac{1}{m},$$
where $\mathcal D_A(x,0^m,\mathcal{D}_{n,m})$ represents $A$'s output distribution on input $(x,0^m,y)$ averaged over $y \sim \mathcal{D}_{n,m}$ and $\| \|$ represents total variation distance.
\end{definition}

\begin{definition}
$\mathsf{SampBQP/qpoly}$ is the class of polynomially-bounded sampling problems $S=\{ \mathcal D_x \}_{x\in \bin^*}$ for which there exists a polynomial-time quantum algorithm $Q$, a polynomial $p(n,m)$, and an infinite list of advice states $\{ \ket{\psi_{n,m}}\}_{n,m \ge 1}$, where $\ket{\psi_{n,m}}$ is on $p(n,m)$ qubits, such that for all $x$ and all $m$,
$$ \| \mathcal D_Q(x,0^m,\ket{\psi_{n,m}}) - \mathcal D_x \| \le \frac{1}{m},$$
where $\mathcal D_Q(x,0^m,\ket{\psi_{n,m}})$ represents $Q$'s output distribution on input $(x,0^m,\ket{\psi_{n,m}})$ and $\| \|$ represents total variation distance.
\end{definition}

Note that our separations will also hold for the \textit{exact} versions of these sampling classes, but the $\varepsilon$-approximate versions are more robust and seem of greater interest.

Our basic results, proved in Section \ref{SAMP}, are as follows. First, we show that sampling classes are more powerful with randomized advice than with deterministic advice:

\begin{restatable}{theorem}{samppolyrpolyrestate}
\label{samppolyrpoly}
$\mathsf{SampBPP/poly} \ne \mathsf{SampBPP/rpoly}$ and $\mathsf{SampBQP/poly} \ne \mathsf{SampBQP/rpoly}$.
\end{restatable}

To prove Theorem \ref{samppolyrpoly}, we simply choose a probability distribution over $\bin^n$ randomly for each $n$, then appeal to a counting argument.

Second, as a straightforward corollary of Theorem \ref{qpolythm}, we show that quantum advice provides more power than classical advice for sampling problems:

\begin{restatable}{theorem}{samprpolyqpolyrestate}
\label{samprpolyqpoly}
$\mathsf{SampBQP/rpoly} \ne \mathsf{SampBQP/qpoly}$.
\end{restatable}

Theorem \ref{samppolyrpoly} contrasts with the situation for relational problems, where $\mathsf{FBPP/poly} = \mathsf{FBPP/rpoly}$ by Theorem \ref{rpolypolyequal}. \ This is noteworthy because Aaronson \cite{aar:samp} used Kolmogorov complexity to prove a general connection between sampling problems and relational problems. \ This connection had the following implication, among others:

\begin{theorem}[\cite{aar:samp}]
$\mathsf{FBPP}=\mathsf{FBQP}$ if and only if $\mathsf{SampBPP}=\mathsf{SampBQP}$.
\end{theorem}

Yet as we now see, the ``equivalence'' does not force the question of the power of randomized advice to have the same answer for sampling problems that it has for relational problems.

See Figure \ref{sampclassfig} for a complexity class inclusion diagram that summarizes our results about sampling classes.

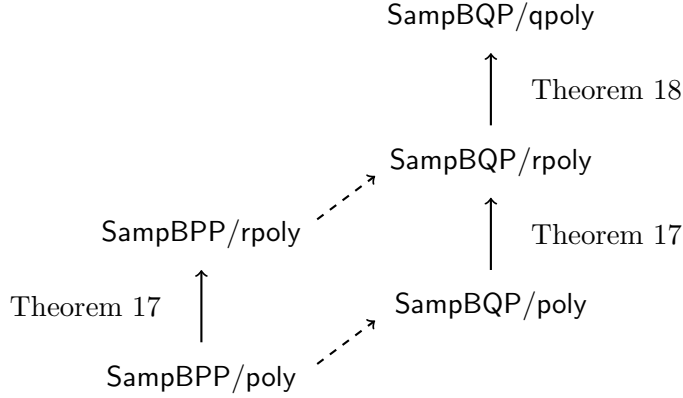
\begin{figure}
\centering
\begin{tikzpicture}[x=5em,y=5em]
    \node at (0,3) {$\mathsf{SampBQP/qpoly}$};
    \draw[thick,->] (0,2.25) -- (0,2.75);
    \node at (0.8,2.5) {Theorem \ref{samprpolyqpoly}};
    \node at (0,2) {$\mathsf{SampBQP/rpoly}$};
    \draw[thick,->] (0,1.25) -- (0,1.75);
    \node at (0.8,1.5) {Theorem \ref{samppolyrpoly}};
    \node at (0,1) {$\mathsf{SampBQP/poly}$ };
    \draw[thick,dashed,->] (-1.2,1.6) -- (-0.8,1.9);
    \node at (-2,1.5) {$\mathsf{SampBPP/rpoly}$};
    \draw[thick,->] (-2,0.75) -- (-2,1.25);
    \node at (-2.8,1) {Theorem \ref{samppolyrpoly}};
    \draw[thick,dashed,->] (-1.2,0.6) -- (-0.8,0.9);
    \node at (-2,0.5) {$\mathsf{SampBPP/poly}$};
\end{tikzpicture}
    \caption{\label{sampclassfig} Relationships among classes of sampling problems considered in this paper. \ \ A solid arrow from $\mathcal{C}_1$ to $\mathcal{C}_2$ indicates strict containment $(\mathcal{C}_1 \subsetneq \mathcal{C}_2)$. \ A dashed arrow indicates a containment $\mathcal{C}_1 \subseteq \mathcal{C}_2$
    that we conjecture to be strict, but a proof of strictness would require a breakthrough in complexity theory.}
\end{figure}







\section{Deterministic, Randomized, and Quantum Advice\label{MAIN}}


We start this section by observing that for relational problems, randomized advice gives no more power than deterministic advice.

\rpolypolyequalrestate*

\begin{proof}
We first prove that $\mathsf{FBPP/rpoly} = \mathsf{FBPP/poly}$. The proof of $\mathsf{FBQP/rpoly} = \mathsf{FBQP/poly}$ is identical but with quantum algorithms in place of randomized algorithms, so we omit it. \ Let $R$ be a relational problem in $\mathsf{FBPP/rpoly}$, decided by an
algorithm $A$. \ Fix an input length $n$ and an $\varepsilon > 0$. \ Let $\mathcal{D}_{n,\varepsilon}$ be the distribution over
advice strings. \ Then for all $x\in\bin^n$, we must have
$$\Pr_{w\sim \mathcal{D}_{n,\varepsilon}} [ (x,A(x,0^{1/\varepsilon},w)) \in R ] \ge 1-\varepsilon.$$

In our $\mathsf{FBPP/poly}$ simulation, we'll take (say) $k=100n/\varepsilon^2$
independent samples $w_1,\ldots,w_k$ from $\mathcal{D}_{n,\varepsilon / 2}$ as the advice. \ Given an
input $x$, we'll then just pick $i\in \{1,\ldots,k\}$ uniformly at random and output $A(x,0^{2/\varepsilon},w_i)$. \ By Hoeffding's inequality, we have that for any fixed $x \in \{0,1\}^n$,
\[
\Pr_{w_1,\ldots,w_k}\left[
\Pr_{i}\left[A(x, 0^{2/\varepsilon}, w_i) \in R \right] < 1 - \varepsilon
\right]
\le
\exp\left(-k \varepsilon^2 / 2\right).
\]

Hence, by a union bound over all $x \in \{0,1\}^n$, there exists some choice of $w_1,\ldots,w_k$ that allows the $\mathsf{FBPP/poly}$ simulation to succeed with probability at least $1 - \varepsilon$ on every $x$.

Lastly, we also have $\mathsf{FBPP/rpoly} = \mathsf{FP/rpoly}$, since the randomized advice to an $\mathsf{FP}$ machine can include as many uniformly random bits as are needed to simulate any desired $\mathsf{FBPP}$ machine.
\end{proof}


We now prove the unconditional separation between $\mathsf{FBQP}$ with quantum advice and $\mathsf{FBQP}$ with classical advice.


\qpolythmrestate*

\begin{proof}
From Theorem \ref{rpolypolyequal}, it suffices to show that $\mathsf{FBQP/qpoly} \neq \mathsf{FBQP/poly}$. \ 
Let $F=\{f_n\}_{n\ge 1}$ be an infinite family of Boolean functions, with
$$ f_n : \bin^n \rightarrow \bin. $$
Then we define the following relation problem:
$$R_F = \{ (x,(y,b)) : x,y \in \bin^n, b\in \bin, f_n(y)\oplus f_n(y\oplus x)=b \}.$$
In other words, given an input $x\in \bin^n$, the problem is to output another string $y\in \bin^n$, along with a bit $b$, such that $f_n(y)$ and $f_n(y\oplus x)$ XOR to $b$.

We first show that, for all $F$, this problem is in $\mathsf{FBQP/qpoly}$. \ The quantum advice state is simply
$$\ket{\psi_n} := \frac{1}{\sqrt{2^n}} \sum_{y \in \bin^n} (-1)^{f_n(y)} \ket{y}.$$
Given an input $x\in\bin^n$, along with $\ket{\psi_n}$, the algorithm is now as follows. \ If $x=0^n$, then just output $(y,0)$ for any $y\in\bin^n$. \ Otherwise, first find a matrix $A\in \mathbb{F}_2^{(n-1)\times n}$ whose nullspace is $\{0,x\}$. \ Then map $\ket{\psi_n}$ to
$$ \frac{1}{\sqrt{2^n}} \sum_{y \in \bin^n} (-1)^{f_n(y)} \ket{y} \ket{Ay} $$
and measure the $\ket{Ay}$ register in the computational basis, to reduce the $\ket{y}$ register to the form
$$ \frac{ (-1)^{f_n(y)} \ket{y} + (-1)^{f_n(y\oplus x)} \ket{y\oplus x} }{\sqrt{2}} $$
for some $y$. \ Then measure the above state in the $\{ \ket{y} \pm \ket{y\oplus x} \}$ basis, to learn the relative phase $b:=f(y)\oplus f(y\oplus x)$. \ Finally, output $y,b$. \ This algorithm succeeds with certainty for every $x$.

By contrast, Theorem \ref{bjkthm_distributional} implies that, with probability $1$ over the choice of $F$, the problem is not in $\mathsf{FBQP/poly}$, or indeed in $\mathsf{FBQP/rpoly}$. \ For each possible input $x\ne 0^n$ gives rise to a matching $\mathcal M_x := \{ (y,y\oplus x) \; | \; y \in \bin^n \}$ on $\bin^n$, and these matchings are pairwise edge-disjoint. \ So, if we imagine that Alice holds the truth table of a random Boolean function $f_n$, consisting of $N=2^n$ bits, while Bob holds a random index $x$ of the matching, we find that Alice must send $\Omega(\sqrt{N})=\Omega(2^{n/2})$ classical bits to Bob to allow him to satisfy the relation $R_F$ with a success probability of at least $7/8$.

In the actual problem, of course, the function $f_n$ is fixed for each $n$, rather than chosen by an Alice, and the $\mathsf{FBQP/poly}$ algorithm $Q$'s behavior depends on the $f_n$'s via the classical advice, rather than a message from Alice. \ Given a choice of $F$, let $a_{F,n,m} \in \{0,1\}^{\operatorname{poly}(n,m)}$ be the advice string for inputs of length $n$ with error $1/m$. \ Then in order for $Q$ to be correct on $x \in \{0,1\}^n$, we require that for all $m$,
\[
\Pr \left[(x, Q(x, 0^m, a_{F,n,m})) \in R_F\right] \ge 1 - 1/m.
\]
If we imagine that $F=\{f_n\}_{n\ge 1}$ is chosen uniformly at random, then we can bound the probability that $Q$ satisfies this condition on all inputs of length $n$, i.e.
\begin{align*}
    &\Pr_F\left[\forall x \in \{0,1\}^n : \Pr\left[(x, Q(x, 0^m, a_{F,n,m})) \in R_F\right] \ge 1 - 1/m \right]\\
    \le &\Pr_{F,x \sim \{0,1\}^n}\left[\Pr \left[(x, Q(x, 0^m, a_{F,n,m})) \in R_F\right] \ge 1 - 1/m \right]\\
    \le &\frac{m}{m-1} \Pr_{F,x \sim \{0,1\}^n} \left[(x, Q(x, 0^m, a_{F,n,m})) \in R_F\right],
\end{align*}
where the last line uses Markov's inequality. \ 
Choose $m = 16$, so that $a_{F,n,m}$ is a string of length $\operatorname{poly}(n) = o\left(2^{n/2}\right)$. \ Then combining the above bound with Theorem \ref{bjkthm_distributional} implies that
\[
\Pr_F\left[\forall x \in \{0,1\}^n : \Pr \left[(x, Q(x, 0^m, a_{F,n,m})) \in R_F\right] \ge 1 - 1/m \right] \le \frac{16}{15} \cdot \frac{7}{8} = \frac{14}{15}
\]
for all sufficiently large $n$. \ Moreover, this probability is independent for each $n \in \mathbb{N}$, because each $f_n$ is chosen independently, so the overall probability that any choice of advice allows $Q$ to compute $R_F$ is at most $\prod_{n=1}^\infty 14/15 = 0$. \ This is to say that a uniformly random $F$ satisfies $R_F \not\in \mathsf{FBQP/poly}$ with probability $1$.
\end{proof}

Note that, in the proof of $R_F \not\in \mathsf{FBQP/poly}$, we nowhere needed the fact that the algorithm was an efficient quantum algorithm (i.e., $\mathsf{FBQP}$), but only that the algorithm succeeds with bounded error. \ Hence we can conclude more generally that $R_F \not\in \mathsf{F}\mathcal{C}\mathsf{/poly}$ for uniform complexity classes $\mathcal{C}$ with arbitrarily large computational power,
such as $\mathsf{PSPACE}$, $\mathsf{EXP}$, $\mathsf{BPEXP}$, $\mathsf{R}$, and so on. \ We additionally get $R_F \not\in \mathsf{FBQP/rpoly}$, because $\mathsf{FBQP/rpoly} = \mathsf{FBQP/poly}$.\footnote{Alternatively, $R_F \not\in \mathsf{FBQP/rpoly}$ can be shown directly by a small modification of the above proof: simply replace the advice string $a_{F,n,m}$ with a sample from an advice distribution. \ This works because Theorem \ref{bjkthm_distributional} lower-bounds \textit{randomized} one-way communication complexity, not just deterministic, by Yao's principle \cite{yao:lemma}.} \ On the other hand, we cannot say that $R_F \not\in \mathsf{F}\mathcal{C}\mathsf{/rpoly}$ for any $\mathcal{C}$, because of the way the success conditions of certain complexity classes interact with randomized advice: as an example, $\mathsf{PostBPP/rpoly}=\mathsf{ALL}$, and so a reasonably defined relational analogue $\mathsf{FPostBPP/rpoly}$ certainly \textit{would} contain $R_F$.

It is interesting to ask just how efficient we can make the quantum algorithm of Theorem \ref{qpolythm}. \ We describe how to implement the measurement on $\ket{\psi_n}$ via a simpler circuit, without the need to compute the matrix-vector multiplication $Ay$. \ We claim the following: first, the quantum circuit for measuring $\ket{\psi_n}$ and learning the output string $y,b$ can be taken to be a stabilizer circuit. \ Second, this stabilizer circuit has $O(n)$ size and can be constructed in $O(n)$ time.

To see why, suppose for example that the input $x$ is $001111$. \ Suppose we measure $\ket{\psi_n}$ according to the circuit in Figure \ref{cktfig}, and get the result $z = z_1z_2z_3z_4z_5z_6$. \ We claim that this measurement result corresponds to collapsing the input state to
\[\frac{\ket{y} + (-1)^b \ket{x \oplus y}}{\sqrt{2}},\]
where $y = z_1z_2z_3z_4z_50$ and $b = z_6$. \ The easiest way to see why is to consider the resulting state when we apply the inverse circuit to $\ket{z}$.

For a general $x$ of Hamming weight $k \ge 1$, we choose an arbitrary $i$ for which $x_i = 1$, and let qubit $i$ play the role of measuring $b$. \ The circuit will consist of $k - 1$ CNOT gates between qubit $i$ and the other qubits $j$ for which $x_j = 1$, followed by a single Hadamard gate on qubit $i$ to measure $b$.


\begin{figure}
\centering
\[
\Qcircuit @C=1.2em @R=0.7em {
\lstick{x_1 = 0 \quad}
& \qw
& \qw
& \qw
& \qw
& \meter
& \rstick{z_1}
\\
\lstick{x_2 = 0 \quad}
& \qw
& \qw
& \qw
& \qw
& \meter
& \rstick{z_2}
\\
\lstick{x_3 = 1 \quad}
& \targ
& \qw
& \qw
& \qw
& \meter
& \rstick{z_3}
\\
\lstick{x_4 = 1 \quad}
& \qw
& \targ
& \qw
& \qw
& \meter
& \rstick{z_4}
\\
\lstick{x_5 = 1 \quad}
& \qw
& \qw
& \targ
& \qw
& \meter
& \rstick{z_5}
\\
\lstick{x_6 = 1 \quad}
& \ctrl{-3}
& \ctrl{-2}
& \ctrl{-1}
& \gate{H}
& \meter
& \rstick{z_6}
\\
}
\]
\caption{A circuit for measuring the state $\ket{\psi_n}$ in Theorem \ref{qpolythm}. \ Note that the circuit will depend on the input $x$; an example with $x=001111$ is shown. \ First, all qubits $i$ such that $x_i=0$ are measured in the computational basis. \ Next, the qubits $i$ such that $x_i=1$ are measured to determine two things: (i) a computational basis state modulo a NOT gate being applied to each qubit; and (ii) the relative phase between those two basis states, one with the NOT gates applied and the other without (this is the purpose of the sole Hadamard gate).}
\label{cktfig}
\end{figure}
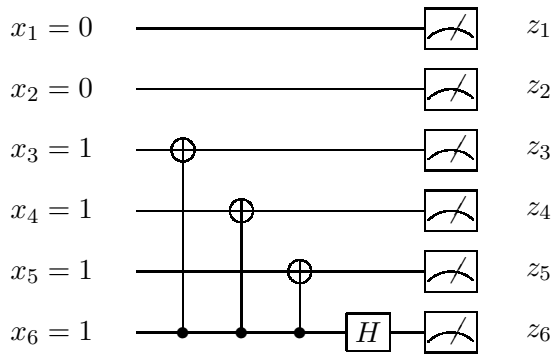

One more comment: our proof of Theorem \ref{qpolythm} was nonconstructive, in the sense that we did not exhibit any particular $F$ such that $R_F \not\in \mathsf{FBQP/poly}$, but merely used counting to show that a random $F$ works with probability $1$. \ Thus, it is natural to wonder whether we could find an ``explicit'' $F$, say $F\in\mathsf{FPSPACE}$, such that $R_F\not\in \mathsf{FBQP/poly}$ under a plausible hardness assumption. \ We do not know, but we would like to observe that for the \textit{promise problem} versions of these classes---namely $\mathsf{PromiseBQP/poly}$ and $\mathsf{PromiseBQP/qpoly}$---general principles imply (perhaps surprisingly) that, if there is any separation at all, then the separation can be witnessed ``explicitly'':

\begin{proposition}
Suppose $\mathsf{PromiseBQP/poly} \ne \mathsf{PromiseBQP/qpoly}$. \ Then there is a $\mathsf{PromisePP}$ problem in $\mathsf{PromiseBQP/qpoly}$ but not in $\mathsf{PromiseBQP/poly}$.
\label{explicit}
\end{proposition}
\begin{proof}
We prove the contrapositive. \ Suppose
$$ \mathsf{PromisePP} \cap \mathsf{PromiseBQP/qpoly} \subseteq \mathsf{PromiseBQP/poly}. $$
Let $\Pi=(\Pi_Y,\Pi_N)$ be a promise problem in $\mathsf{PromiseBQP/qpoly}$. \ Aaronson \cite{aar:adv} showed that $\mathsf{PromiseBQP/qpoly} \subseteq \mathsf{PromisePP/poly}$. \ Let $\{w_n\}_{n\ge 1}$ be the  advice strings for the $\mathsf{PromisePP/poly}$ machine. \ Then we define a new promise problem $\Pi'$, whose yes-instances have the form $(x,w_n)$ for $x\in\Pi_Y \cap \bin^n$, and whose no-instances have the form $(x,w_n)$ for $x\in\Pi_N \cap \bin^n$. \ Clearly $\Pi'\in \mathsf{PromisePP}$. \ We also have $\Pi'\in \mathsf{PromiseBQP/qpoly}$, since we can just ignore $w_n$. \ By assumption, then, $\Pi'\in \mathsf{PromiseBQP/poly}$. \ But moving $w_n$ back to the advice, this means that $\Pi\in \mathsf{PromiseBQP/poly}$ as well. \ Therefore $\mathsf{PromiseBQP/qpoly} = \mathsf{PromiseBQP/poly}$.
\end{proof}

So in particular, under a plausible complexity assumption, namely $\mathsf{BQP/poly}\ne \mathsf{BQP/qpoly}$, there \textit{is} an explicit problem in $\mathsf{FBQP/qpoly}$ but not in $\mathsf{FBQP/poly}$---for observe that promise problems are a special case of relational problems.

Of course, for relational problems we would like to do better than this observation, by constructing an explicit problem in $\mathsf{FBQP/qpoly} \setminus \mathsf{FBQP/poly}$ under a ``standard'' hardness assumption, one about complexity classes like $\mathsf{EXP}$ or $\mathsf{PSPACE}$ or $\mathsf{P/poly}$.

\section{Toward Quantum Information Supremacy\label{SUPREMACY}}

The extreme simplicity of the circuit in Figure \ref{cktfig} to measure the advice state $\ket{\psi_n}$---namely, a linear number of 1- and 2-qubit Clifford gates---raises the question of whether an experiment ``witnessing'' the separation between $\mathsf{FBQP/rpoly}$ and $\mathsf{FBQP/qpoly}$ might be feasible with current technology, and on a large enough scale to be interesting.

As it happens, Kumar, Kerenidis, and Diamanti \cite{kkd} reported an experimental demonstration of the Bar-Yossef-Jayram-Kerenidis Hidden Matching protocol \cite{bjk} in 2018. \ However, their experiment used an optical coherent state with $2^n$ modes, and tiny average photon number per mode, in order to simulate $n$ qubits. \ It therefore didn't directly test the question of whether $n$ qubits (encoded, say, using $n$ entangled particles) require $\exp(n) \gg n$ bits to simulate classically. \ This is the question that we propose to test now.

We are finally in the era of small programmable quantum computers: devices that can run more-or-less arbitrary circuits (subject to locality constraints) on $\sim 100$ qubits and $\sim 1000$ gates, and then extract a detectable signal on measurement. \ Within the past few years, these devices have been used to assert the milestone of ``quantum computational supremacy''---that is, a clear advantage over currently-available classical algorithms and hardware---for contrived tasks such as Random Circuit Sampling and BosonSampling (see, e.g., \cite{arute,ustc}). 

Notably, these quantum supremacy tasks are validated using, e.g., Google's \textit{Linear Cross-Entropy Benchmark} \cite{arute}, which lets us see them as literally \textit{relational problems}. \ In Random Circuit Sampling, for example, we are given as input a classical description of an $n$-qubit quantum circuit $C$. \ We are then asked to output any distinct strings $s_1,\ldots,s_k\in\bin^n$ that satisfy an inequality such as
$$\sum_{i=1}^k |\langle s_i | C | 0^n \rangle|^2 \ge \frac{1.002 k}{2^n}.$$
When, say, $k\approx n$, the above yields a relational problem in $\mathsf{FBQP}$ that is plausibly conjectured not to be in $\mathsf{FBPP}$ (see, e.g., Aaronson and Gunn \cite{agunn}).

Of course, any conjecture of this sort rests on unproved computational hardness assumptions: if nothing else, then $\mathsf{P}\ne\mathsf{PSPACE}$! \ Alas, ``standard'' hardness assumptions have not sufficed here. \ The fundamental drawback of current quantum supremacy experiments is that, with each concession that we need to make to experimental reality---for example, depolarizing noise, photon losses, a limited number of qubits (to keep the classical verification feasible), limited circuit depth (to control the noise), limited qubit connectivity, etc.---the relevant hardness assumptions move to shakier and shakier ground. \ Furthermore, the worry about classical spoofing is far from hypothetical! \ Classical algorithms for simulating noisy random quantum circuits \textit{have} improved, both in theory and in practice (see, e.g., \cite{gaoetal,panchenzhang,agllv})---if not enough to kill the current claims of quantum supremacy outright, then enough to call them into reasonable doubt.

Thus, wouldn't it be great if a meaningful quantum supremacy experiment could be designed based on \textit{no} unproved hardness assumptions? \ Such an experiment might, for example, try to falsify the hypothesis that every ``realistic'' entangled state of $n$ qubits is secretly describable using $p(n)$ classical bits, for some small polynomial $p$, whether due to noise or experimental limitations or even a breakdown of quantum mechanics itself. \ Note that this marks a fundamental difference compared to existing experiments based on (say) Bell inequality violations---whereas the Bell/CHSH experiments test the nonlocal nature of quantum correlations, they in no way test the exponential dimensionality of Hilbert space. \ As far as we know, the experiments we propose here would be the first directly to test the latter without relying on any unproved computational assumptions.

In such an experiment, we might first prepare a random $n$-qubit entangled state $\ket{\psi}$, then measure $\ket{\psi}$ in a basis $B$ chosen randomly and ``on the fly''---just like Alice's and Bob's measurement bases in the Bell/CHSH experiment are ideally chosen when the entangled photons are already in flight. \ We would repeat this process many times, with a new random basis $B$ each time, collect statistics on the measurement outcomes, and then argue that no $p(n)$-bit classical digest of $\ket{\psi}$ could possibly have allowed those statistics to be reproduced.

This is exactly what we suggest to do, by leveraging the unconditional separation between $\mathsf{FBQP/rpoly}$ and $\mathsf{FBQP/qpoly}$, which in turn is based on the unconditional separation between randomized and quantum one-way communication complexities. \ Note that such an experiment would almost certainly be impractical with current hardware, if the required measurements on $\ket{\psi}$ were complicated ones. \ In practice, then, it is essential that we have not merely an exponential separation in one-way communication complexities for a relational problem, but one wherein Bob's measurements are ``simple''---which is the content of $\mathsf{FBQP/rpoly}\ne\mathsf{FBQP/qpoly}$.

The detailed consideration of such an experiment is beyond this paper's scope. \ Briefly, though, laying the groundwork for this experiment would involve refining and improving the $\mathsf{FBQP/rpoly}\ne\mathsf{FBQP/qpoly}$ separation in several ways. \ Firstly, one would want a separation between randomized and quantum advice length that was as \textit{quantitatively tight} as possible, and that also included concrete bounds for particular small numbers of qubits $n$, such as $20$ or $30$. \ Secondly, one would want to account for the complexity of \textit{preparing} the advice state $\ket{\psi}$: for example, what if we chose $\ket{\psi}=C\ket{0^n}$, where $C$ is a random quantum circuit with $m$ gates? \ Thirdly, one would, if possible, want the measurements of $\ket{\psi}$ to be \textit{even simpler} than the one from Figure \ref{cktfig}, and lower-depth: for example, could one even measure each qubit separately? \ Fourthly, one would want an analysis that accounted for $\ket{\psi}$ being \textit{extremely noisy}---as it will be, in any near-term implementation---and that carefully quantified the experimental resources needed to achieve a clear separation between randomized and quantum advice length even in the teeth of the noise. \ See Section \ref{OPEN} for further discussion of these challenges, especially the quantitative tightness one.

Once all the challenges are taken into account, the separation between randomized and quantum information achievable with current devices might be rather modest, as it was in the earlier work by Kumar, Kerenidis, and Diamanti \cite{kkd}. \ For example, perhaps it will be possible to perform an experiment with $n\approx 20$ qubits, using $2$-qubit gates of $99.8\%$ fidelity or whatever, to verify that any secretly classical description of the qubits' state (even a probabilistic description) would need at least $\sim 100$ bits, in order to explain the observed success at measuring the state to solve a relational problem. \ In our view, though, this would already be a historic result, sufficient to disturb the certainty of those who regard the vastness of Hilbert space as just a theoretical fiction. \ We hope to address some of the challenges in future work.

\section{The Power of \textsf{FBPP}\label{FBPPSEC}}

In this section we show several senses in which $\mathsf{FBPP}$ behaves differently from its decision-problem counterpart. \ We start with an unconditional separation between $\mathsf{FP}$ and $\mathsf{FBPP}$.

\fpneqfbpprestate*

\begin{proof}
Recall the definition of Levin's time-bounded Kolmogorov complexity: given a string $y$,
$$ \operatorname{Kt}(y) := \min_{P : P()=y} (|P| + \log_2 t(P)): $$
that is, we minimize the length of a program $P$ (in some fixed programming language) plus the logarithm of $P$'s runtime, over all programs $P$ that output $y$ given a blank input. \ Now consider the relation
$$ R = \{ (x,y) : |x|=|y|, \operatorname{Kt}(y)\ge \frac{|y|}{2} \}. $$
We first show that $R\in \mathsf{FBPP}$. \ Given an input of length $n$, the strategy depends on the allowed error probability $\varepsilon$. \ If $\varepsilon \ge \frac{1}{2^{n/2}}$, then we can simply output a uniformly random $y\in\bin^n$; a counting argument will then imply that $\operatorname{Kt}(y)\ge \frac{n}{2}$ with probability at least $1-\varepsilon$. \ If, on the other hand, $\varepsilon < \frac{1}{2^{n/2}}$, then we can use brute force to find and output the lexicographically first string $y\in\bin^n$ such that $\operatorname{Kt}(y)\ge \frac{n}{2}$. \ This takes time exponential in $n$, so polynomial in $\frac{1}{\varepsilon}$.

Next we show that $R\not\in \mathsf{FP}$. \ Let $A$ be a deterministic algorithm with polynomial running time $p(n)$; then for all $n$, we clearly have
$$\operatorname{Kt}(A(0^n))\le |A| + \log_2 n + \log_2 p(n).$$
But the above is less than $\frac{n}{2}$ for all sufficiently large $n$, which proves that $A$ cannot compute $R$.
\end{proof}

Next we show that the unconditional separation still holds if $\mathsf{FP}$ and $\mathsf{FBPP}$ both have polynomial-sized advice.

\fpneqfbpppolyrestate*
\begin{proof}
For each $n$ and $x\in \bin^n$, we choose a subset $S_x \subset \bin^n$ uniformly at random and independently subject to $|S_x| = 2^{n/2}$. \ We then define the following relational problem $R$:
$$ R = \{ (x,y) : |y|=|x|, y\not\in S_x \}.$$
In other words: given as input an $n$-bit string $x$, the problem is to output an $n$-bit string $y$ that is not in $S_x$.

We first claim that $R\not\in \mathsf{FP/poly}$, with probability $1$ over the choice of $S_x$’s. \ This is because, for any fixed $\mathsf{FP/poly}$ algorithm $C$ and any $n$, we have
$$\Pr_{\{S_x\}_{x\in\bin^n}}[(x,C(x))\in R \; \forall x\in\bin^n] \le \left(1-\frac{1}{2^{n/2}}\right)^{2^n} = \frac{1}{\exp(2^{n/2})}, $$
which remains small even after we take a union bound over all possible $C$'s.

By contrast, we claim that $R \in \mathsf{FBPP/poly}$. \ The algorithm is as follows: if $\varepsilon \ge 1/2^{n/2}$, then just output a uniformly random $y\in\bin^n$. \ If, on the other hand, $\varepsilon < 1/2^{n/2}$, then the advice string can have size $2^n n = \operatorname{poly}(n,1/\varepsilon)$, and can therefore just provide a giant list containing some $y\not\in S_x$ for each possible $x\in\bin^n$.
\end{proof}

Finally, we give strong evidence that $\mathsf{FBPP} \not\subset \mathsf{FP/poly}$.


\fppolythmrestate*

\begin{proof}
We use an idea of Buhrman and Torenvliet \cite{buhrtor}, which is in turn based on the $\mathsf{IP}=\mathsf{PSPACE}$ theorem and conditional time-bounded Kolmogorov complexity.\footnote{Buhrman and Torenvliet \cite{buhrtor} used these ideas to prove that $\mathsf{PSPACE}\subseteq\mathsf{NP}^{R_s^{CS}}$, where $R_s^{CS}$ is an oracle to decide whether a given string has maximal space-bounded conditional Kolmogorov complexity. \ We, by contrast, are interested in the nonuniform complexity of \textit{relational} problems. \ Outputting strings of large time-bounded conditional Kolmogorov complexity is a convenient relational problem for proving the implication we want.}

Given strings $x$ and $y$, we define $$\operatorname{Kt}(y|x) := \min_{P : P(x)=y} (|P| + \log_2 t(P,x)),$$
or the \textit{time-bounded Kolmogorov complexity of $y$ conditioned on $x$}, to be the minimum, over all programs $P$ (in some fixed programming language) such that $P(x)=y$, of the bit-length of $P$ plus the log of its runtime on input $x$.

We now define the following relation:
$$ R^* := \{ ((x,0^n),y) : y\in\bin^n, \operatorname{Kt}(y|x) \ge n/2 \}. $$
In other words, given as input $x$ (which could have some arbitrary length $m=\operatorname{poly}(n)$) and $0^n$, the problem is to output an $n$-bit string $y$ that cannot be computed too quickly by any short, deterministic program given $x$.

Our first claim is that $R^*\in \mathsf{FBPP}$. \ The argument is the same as in the proof of Theorem \ref{fpneqfbpp}; the fact that we condition on $x$ makes no difference. \ And so---we save this fact for later---if $\mathsf{FBPP} \subset \mathsf{FP/poly}$, then $R^*$ is also decided by some polynomial-size family of circuits $\{C_{m,n}\}_{m,n\ge 1}$.

We now recall the relevant facts about the proof of $\mathsf{IP}=\mathsf{PSPACE}$ (see \cite{shamir} for details). \ Let $\phi$ be an instance of $TQBF$, the canonical $\mathsf{PSPACE}$-complete problem. \ Then to verify that $\phi\in TQBF$, the verifier engages the prover in an $n$-round conversation about a certain complicated (but polynomial-sized) arithmetic expression of the form
$$ P = \sum_{x_1\in\bin} R_{x_1} \prod_{x_2\in\bin} R_{x_1} R_{x_2} \sum_{x_3\in\bin} R_{x_1} R_{x_2} R_{x_3} \cdots \prod_{x_n\in\bin} R_{x_1} \cdots R_{x_n} \varphi(x_1,\ldots,x_n)$$
over the finite field $\mathbb{F}_q$, where we take $q$ to be a prime such that $q> 16^{n^2}$. \ The expression involves three types of quantifiers over variables: sums, products, and so-called \textit{degree reduction operators} (these are the $R_{x_i}$'s).

The expression $P$ is constructed by carefully arithmetizing $\phi$ to maintain the following properties:
\begin{enumerate}
\item $P=1$ if $\phi\in TQBF$ while $P=0$ if $\phi\not\in TQBF$.
\item For any $t$, if we substitute field values $r_1,\ldots,r_{t-1}\in \mathbb{F}_q$ at the first $t-1$ quantifiers in an appropriate way, remove the $t^{th}$ quantifier, and keep in place everything to the right of the $t^{th}$ quantifier, then we are left with a univariate polynomial $h_t(x_i)$ of degree at most $\operatorname{poly}(n)$ over $\mathbb{F}_q$, where $x_i$ is the variable that appears in the $t^{th}$ quantifier. \ (The whole purpose of the degree reduction operators is to ensure this.)
\end{enumerate}

The conversation proceeds in $T\le n^2+1$ rounds, one for each quantifier. \ At round $t$, the prover sends the verifier a univariate polynomial $g_t : \mathbb{F}_q\rightarrow \mathbb{F}_q$, and makes the crucial claim that $g_t=h_t$ as polynomials over $\mathbb{F}_q$, where $h_t$ is the univariate polynomial discussed previously, which depends on the random finite field values $r_1,\ldots,r_{t-1}\in \mathbb{F}_q$ chosen by the verifier in the previous rounds. \ After applying some preliminary checks, the verifier then tests the prover's claim by choosing a new $r_t \in \mathbb{F}_q$ uniformly at random and sending it to the prover, and the conversation continues. \ Finally, at the very last round, the verifier can check $g_T$ directly against a polynomial obtained by arithmetizing $\phi$.

The key fact is that, by the Fundamental Theorem of Algebra, if $g_t\neq h_t$ as polynomials, then $g_t(r)$ and $h_t(r)$ can coincide on at most $\max\{\deg(g_t),\deg(h_t)\}=n^{O(1)}$ values of $r$. \ And these are the only values of $r\in \mathbb{F}_q$ that can cause the protocol to fail at round $t$ (in the sense that the verifier will now accept even though $\phi\not\in TQBF$).

For technical reasons to be explained later, let $s\in\bin^{n^{O(1)}}$ be a polynomial-sized string that is chosen uniformly at random and then fixed.

We now make the following observation: \textit{in place of a uniformly random $r_t \in \mathbb{F}_q$, the verifier could send any $r_t \in \mathbb{F}_q$ such that $\operatorname{Kt}(r_t | s,\phi,r_1,\ldots,r_{t-1},g_t) \ge 2n^2$.}

To see why, consider a ``bad'' $r_t$: that is, one such that $g_t(r_t)=h_t(r_t)$, even though $g_t\neq h_t$ as polynomials. \ As we said, there can be at most $n^{O(1)}$ such bad $r_t$'s. \ Furthermore, we claim that the complete list of bad $r_t$'s can be generated in $2^T n^{O(1)}$ time, given $\phi$ and $s,r_1,\ldots,r_{t-1}$ and $g_t$ as input, with overwhelming probability over the choice of $s$.

To generate the list, we first compute $h_t$ explicitly as a polynomial over $\mathbb{F}_q$, by simply ``brute-forcing'' every sum and product over a variable $x_i\in\bin$ and every degree reduction operator that appears to the right of the $t^{th}$ quantifier. \ This takes time $2^{T-t} n^{O(1)}$, since we pick up a factor of $2$ for every quantifier that needs to be brute-forced. \ We next factor the polynomial $g_t(r)-h_t(r)$ over the finite field $\mathbb{F}_q$---for example, by using the randomized algorithm due to Berlekamp \cite{berlekamp:polyfac}, which runs in $\operatorname{poly}(n,\log q)=\operatorname{poly}(n)$ time, and which fails with probability at most $1/2^{p(n)}$, where $p$ is a polynomial that we can make as large as needed by choosing a large enough randomness string $s$ to feed to Berlekamp's algorithm. \ Finally, from the degree-$1$ irreducible factors of $g_t-h_t$, we extract the solutions $r\in \mathbb{F}_q$ to $g_t(r)=h_t(r)$.

Thus, letting $\Pi$ be a program that does the above, for any bad $r_t$, we have
$$ \operatorname{Kt}(r_t | s,\phi,r_1,\ldots,r_{t-1},g_t) \le |\Pi| + \log_2 \deg(g_t-h_t) + \log_2(2^T n^{O(1)}),$$
which is less than $2T$ for all large enough $n$ and $T\approx n^2$. \ And so, taking the contrapositive, if $r_t$ has conditional time-bounded Kolmogorov complexity at least $2T$, then it cannot be bad.

But if we set $x:=\langle s,\phi,r_1,\ldots,r_{t-1},g_t\rangle$, then the problem of finding an $r_t$ such that $$\operatorname{Kt}(r_t | s,\phi,r_1,\ldots,r_{t-1},g_t) \ge 2T$$ can be solved by finding a $y\in\bin^{4n^2}$ such that $((x,0^{4n^2}),y)$ is in the relation $R^*$. \ And we said that, by the assumption $\mathsf{FBPP} \subset \mathsf{FP/poly}$, there is a polynomial-size circuit family $\{C_{m,n}\}_{m,n\ge 1}$ that does this.

Hence we can decide $TQBF$ in $\mathsf{NP/poly}$, as follows. \ Given as input a $TQBF$ instance $\phi(x_1,\ldots,x_n)$, the polynomial-sized advice provides a description of an appropriate circuit $C_{m,4n^2}$, along with a hardwired value for the randomness string $s$. \ Given $\phi$ and given this advice, the $\mathsf{NP}$ prover is asked to provide a complete transcript for the $\mathsf{IP}=\mathsf{PSPACE}$ protocol, \textit{assuming that the verifier generates each of its messages using $C_{m,4n^2}$}. \ Finally, the $\mathsf{NP}$ verifier checks each step in this transcript, using $C_n$ to make sure that the prover computed the $\mathsf{IP}$ verifier's messages correctly.

By the reasoning above, this derandomization of $\mathsf{IP}=\mathsf{PSPACE}$ is sound: any failure would imply that $C_{m,4n^2}$ had generated a message of small time-bounded conditional Kolmogorov complexity. \ Or more precisely, this is true with overwhelming probability over the choice of $s$, which means that there must exist fixed $s$'s that work when hardwired into the advice. \ Therefore $\mathsf{PSPACE} \subset \mathsf{NP/poly}$ as claimed.
\end{proof}

We conclude this section by observing a barrier to any unconditional proof of $\mathsf{FBPP} \not\subset \mathsf{FP/poly}$, as it would lead to new circuit lower bounds.

\fbppupperboundrestate*

\begin{proof}
We prove the first part of the theorem; the second part is an immediate consequence. \ Let $R \in \mathsf{FBPP}$. \ For simplicity, suppose there exists a polynomial $p(|x|)$ such that for every $(x, y) \in R$, $|y| = p(|x|)$ (which can always be assumed under a suitable efficient encoding).

Let $A(x,0^{1/\varepsilon})$ be the probabilistic algorithm for computing $R$ with probability at least $1 - \varepsilon$ in time $\operatorname{poly}(|x|, 1/\varepsilon)$. \ Fix $\varepsilon(|x|) = 4^{-p(|x|)}$.

Let $\Pi$ be the following promise problem of, given an input $(x, z)$, to decide whether:
\begin{itemize}
\item[(YES)] With probability at least $2/3 \cdot 3^{-|z|}$, the prefix of $A(x, 0^{1/\varepsilon(|x|)})$ is $z1$, or
\item[(NO)] With probability at most $1/3 \cdot 3^{-|z|}$, the prefix of $A(x, 0^{1/\varepsilon(|x|)})$ is $z1$,
\end{itemize}
promised that one of (YES) or (NO) is the case.

Observe that $\Pi \in \mathsf{PromiseBPEXP}$: the algorithm runs $A(x, 0^{1/\varepsilon(|x|)})$ on (say) $100^{|z|}$ independent random strings, and outputs YES or NO depending on whether than more than a $1/2 \cdot 3^{-|z|}$ fraction of the strings begin with $z1$. \ A Chernoff bound guarantees that the algorithm is correct with high probability, so long as $(x, z)$ is in the promise.

Next, we claim that $R \in \mathsf{FP}^\Pi$. \ The algorithm for outputting $(x, y) \in R$ is as follows. \ Let $z_0 = \emptyset$. \ For each $i \in p(|x|)$, compute $z_i = z_{i-1}\Pi(x,z_{i-1})$. \ That is, we obtain $z_i$ by appending $1$ to $z_{i-1}$ if the $\Pi$ oracle answers YES on $(x, z_{i-1})$, and by appending $0$ otherwise. \ Finally, output $z_{p(|x|)}$.

The correctness of the algorithm follows by observing that after step $i$ of the algorithm, $A(x, 0^{1/\varepsilon(|x|)})$ has probability at least $3^{-i}$ of outputting a string that starts with $z_i$. \ The proof is by induction on $i$: either $(x, z_{i-1})$ satisfies the promise, in which case $z_i = z_{i-1}\Pi(x,z_{i-1})$ appears as a prefix with probability at least $2 \cdot 3^{-i}$, or else both $z_{i-1}0$ and $z_{i-1}1$ appear as a prefix with probability at least $3^{-i}$. \ Then, $y = z_{p(|x|)}$ must be a string that $A(x, 0^{1/\varepsilon(|x|)})$ outputs with probability at least $3^{-p(|x|)}$. \ But since $A$ errs with probability at most $\varepsilon(|x|) = 4^{-p(|x|)}$, we conclude that $(x, y) \in R$.
\end{proof}




\section{Alternative Error Bounds\label{SMALLERROR}}

In this section, we consider some of the consequences of modifying the error bounds in the definition of $\mathsf{FBPP}$.

Define the class $\mathsf{FBPP}_{\operatorname{log}}$ exactly the same way as $\mathsf{FBPP}$, except that now the algorithm is required to take $\operatorname{poly}(n,\log 1/\varepsilon)$ time rather than merely $\operatorname{poly}(n,1/\varepsilon)$. \ In other words, we mandate that the algorithm can reduce the error probability to an exponentially small quantity in polynomial time. \ Clearly $\mathsf{FP} \subseteq \mathsf{FBPP}_{\operatorname{log}} \subseteq \mathsf{FBPP}$.

As we shall see here, the complexity situation for $\mathsf{FBPP}_{\operatorname{log}}$ differs dramatically from that for $\mathsf{FBPP}$. \ First, we observe that $\mathsf{FBPP}_{\operatorname{log}}$ \textit{cannot} be unconditionally separated from $\mathsf{FP}$:

\begin{proposition}
If $\mathsf{P}=\mathsf{NP}$, then $\mathsf{FBPP}_{\operatorname{log}} = \mathsf{FP}$.
\end{proposition}
\begin{proof}
Let $R\in \mathsf{FBPP}_{\operatorname{log}}$, and let $p$ be a polynomial such that $|y|\le p(n)$ for all $(x,y)\in R$ with $|x|\le n$. \ Set $\varepsilon := 1/4^{p(n)}$. \ Then there exists a randomized algorithm $A$ that, given $x\in\bin^n$, outputs a $y\in\bin^{\le p(n)}$ such that $(x,y)\in R$, with success probability at least $1-\varepsilon$, in $\operatorname{poly}(n,\log 1/\varepsilon) = \operatorname{poly}(n)$ time.

This means that, under the assumption $\mathsf{P}=\mathsf{NP}$ (and hence $\mathsf{P}=\mathsf{PH}$), in $\mathsf{FP}$ we can use Stockmeyer approximate counting \cite{stockmeyer} to find a $y\in\bin^{\le p(n)}$ such that (say) $\Pr[A(x)=y]\ge \frac{0.1}{2^{p(n)}}$, which must exist by an averaging argument. \ Such a $y$ must then satisfy $(x,y)\in R$, by the assumption that $A$ succeeds with probability at least $1-\varepsilon$.
\end{proof}

Second, we observe that the analogue of Adleman's Theorem \cite{adleman:bpp} \textit{does} hold for $\mathsf{FBPP}_{\operatorname{log}}$:

\begin{proposition}
$\mathsf{FBPP}_{\operatorname{log}} \subset \mathsf{FP/poly}$.
\label{prop:fbpplogadlemans}
\end{proposition}
\begin{proof}
Set $\varepsilon := 1/4^n$. \ Then
the $\mathsf{FBPP}_{\operatorname{log}}$ machine takes $\operatorname{poly}(n,\log 1/\varepsilon) = \operatorname{poly}(n)$ time, and we can fix as the $\mathsf{FP/poly}$ advice a single randomness string that works for all inputs $x \in \bin^n$, which must exist by the union bound.
\end{proof}

Combining Proposition \ref{prop:fbpplogadlemans} with Theorem \ref{fppolythm} implies that $\mathsf{FBPP}_{\operatorname{log}} \neq \mathsf{FBPP}$, unless the polynomial hierarchy collapses. \ To summarize, then, $\mathsf{FBPP}_{\operatorname{log}}$ behaves more like the decision class $\mathsf{BPP}$ than it does like $\mathsf{FBPP}$.

A different choice would be to consider $\mathsf{FBPP}_{\operatorname{negl}}$, which we define as the class of all polynomially-bounded relations $R\subseteq \bin^{*}\times \bin^{*}$ for which there exists a polynomial-time randomized algorithm that, given $x$, outputs a $y$ such that $(x,y)\in R$ (whenever one exists) with success probability at least $1-\varepsilon(n)$, for some negligible function $\varepsilon(n) = \frac{1}{n^{\omega(1)}}$.

We have the following unconditional result, which supersedes the analogues of Theorems \ref{fppolythm}, \ref{fpneqfbpp}, and \ref{fpneqfbpppoly}:

\begin{theorem}
$\mathsf{FBPP}_{\operatorname{negl}} \not\subset \mathsf{FP/poly}$.
\end{theorem}
\begin{proof}
The relation $R$ that witnesses the separation is the same one from the proof of Theorem \ref{fpneqfbpppoly}, involving a ``bad output set'' $S_x\subset \bin^n$ chosen uniformly at random for each $x\in\bin^n$ subject to $|S_x|=2^{n/2}$. \ We already showed in Theorem \ref{fpneqfbpppoly} that $R\not\in \mathsf{FP/poly}$. \ For $R\in \mathsf{FBPP}_{\operatorname{negl}}$, the algorithm is just to output an $n$-bit string uniformly at random, independent of $x$.
\end{proof}

What is the relationship between $\mathsf{FBPP}_{\operatorname{negl}}$ and $\mathsf{FBPP}$? \ In one direction we have the following:

\begin{proposition}
$\mathsf{FBPP}_{\operatorname{negl}}\not\subset \mathsf{FBPP}$.
\end{proposition}
\begin{proof}
Consider the relation
$$ R = \{ (x,y) : |x|=|y|, \operatorname{K}(y)\ge \frac{|y|}{2}\}, $$
where $\operatorname{K}$ is Kolmogorov complexity. \ We have $R\in \mathsf{FBPP}_{\operatorname{negl}}$ by simply outputting a random string of length $n=|x|$. \ On the other hand, if $R$ had an $\mathsf{FBPP}$ algorithm, then by setting $\varepsilon := 1/4^n$ and then searching for the lexicographically first string that the algorithm output with probability at least (say) $\frac{0.1}{2^n}$, we could deterministically compute an $n$-bit string such that $\operatorname{K}(x)\ge n/2$, which is impossible.
\end{proof}

In the other direction, we leave open whether $\mathsf{FBPP} \subset \mathsf{FBPP}_{\operatorname{negl}}$ or whether the two classes are incomparable. \ Clearly we do have $\mathsf{FBPP}_{\operatorname{log}} \subset \mathsf{FBPP}_{\operatorname{negl}}$, by setting (say) $\varepsilon = 1/2^n$.

\section{Sampling Problems\label{SAMP}}

We now show that the analogue of Theorem \ref{rpolypolyequal} is \textit{false} for sampling problems:


\samppolyrpolyrestate*

\begin{proof}
Like in Theorem \ref{rpolypolyequal}, we prove the separation involving $\mathsf{FBPP}$, as the separation involving $\mathsf{FBQP}$ is completely analogous. \ It suffices to choose some family of nonempty sets $S_n \subset \bin^n$, one for each $n$. \ Then consider the problem of outputting a uniformly random element of $S_n$ on input $x\in \bin^n$. \ This problem is clearly in  $\mathsf{SampBPP/rpoly}$, since we can take the randomized advice itself to be a uniformly random element of $S_n$. \ But if the $S_n$'s are chosen uniformly at random, then a counting argument shows that the problem has probability $0$ of being in $\mathsf{SampBPP/poly}$.
\end{proof}

Lastly, we observe that Theorem \ref{qpolythm} gives rise to a separation of sampling classes with randomized and quantum advice:

\samprpolyqpolyrestate*

\begin{proof}
Consider the problem of sampling from the output distribution of the errorless algorithm that solves the relation problem $R_F$ in the proof of Theorem \ref{qpolythm}. \ This sampling problem is in $\mathsf{SampBQP/qpoly}$. \ On the other hand, if this sampling problem were in $\mathsf{SampBQP/rpoly}$, this would imply $R_F\in \mathsf{FBQP/rpoly}$, because sampling from the distribution within total variation distance $\varepsilon$ would solve $R_F$ with probability $1 - \varepsilon$. \ But this would violate Theorem \ref{qpolythm}.
\end{proof}

To summarize, then, the classes $\mathsf{SampBQP/poly}$, $\mathsf{SampBQP/rpoly}$, and $\mathsf{SampBQP/qpoly}$ are \textit{all} distinct.

\section{Open Problems\label{OPEN}}

We gave an example of a relational problem $R_F$ in $\mathsf{FBQP/qpoly}$ but not in $\mathsf{FBQP/rpoly}$: indeed, one that is easy to solve using $n$ qubits of quantum advice, but requires $\Omega(2^{n/2})$ bits of classical randomized advice. \ While this separation is tight for $R_F$ itself, is there a different relational problem, solvable with $n$ qubits of advice, for which the lower bound on randomized advice reaches its maximum of $\Omega(2^n)$?

We note that Montanaro \cite{montanaro:lb} showed in 2019 that $\Omega(2^n)$ bits of classical advice are needed to perform certain sampling tasks, for which $n$ qubits of quantum advice suffice. \ In the other direction, Gosset and Smolin \cite{gossetsmolin} have shown that, for decision and promise problems solvable with $n$ qubits of quantum advice, $O(2^{n/2})$ bits of classical randomized advice suffice. \ The case of relational problems remains open. \ We conjecture that the answer is closer to $2^n$ than $2^{n/2}$.

As we discussed in Section \ref{SUPREMACY}, the gap between $2^n$ and $2^{n/2}$ would be \textit{extremely} useful to close, as a prerequisite to any possible quantum information supremacy experiment based on the $\mathsf{FBQP/rpoly} \ne \mathsf{FBQP/qpoly}$ separation. \ Other complexity results that would bear directly on such an experiment include:

\begin{enumerate}
\item a lower bound on the amount of classical randomized advice needed to simulate \textit{noisy} quantum advice for some relational problem---say, as a function of the fidelity $\delta$ between the actual advice state and a desired pure state;

\item a quantitative refinement of $\mathsf{FBQP/rpoly} \ne \mathsf{FBQP/qpoly}$ that took into account the circuit complexity of \emph{preparing} the $n$-qubit quantum advice state (which, in practice, would likely have to be much less than $2^n$); and

\item a reproof of $\mathsf{FBQP/rpoly} \ne \mathsf{FBQP/qpoly}$ in which the circuit to measure the quantum advice state was made \textit{as simple as possible}---could it even measure each of the $n$ qubits independently from the rest?
\end{enumerate}

Moving on, is there any sense in which $\mathsf{FBQP/qpoly}$ contains ``more'' problems than $\mathsf{FBQP/rpoly}$---i.e., can the classes be separated by counting the number of problems in each?

Can we separate $\mathsf{FBQP/qpoly}$ from $\mathsf{FBQP/rpoly}$ via an ``explicit'' problem, rather than relying on the probabilistic method? \ More concretely: can we show that, under some plausible hardness assumption, there is a relation in (say) $\mathsf{FBQP/qpoly} \cap \mathsf{FPSPACE}$ that is not in $\mathsf{FBQP/rpoly}$? \ Such a result would create more ``symmetry'' between that separation and our $\mathsf{FP/poly}\ne\mathsf{FP/rpoly}$ separation. \ For the latter, Theorem \ref{fppolythm} gave an explicit relational problem that plausibly realizes the nonconstructively proven separation: namely, the problem of outputting a string of high time-bounded Kolmogorov complexity, conditional on an input string $x$.

Would $\mathsf{FBPP}\subset \mathsf{FP/poly}$ have even stronger consequences than $\mathsf{PSPACE}\subset \mathsf{NP/poly}$, such as $\mathsf{PSPACE}\subset \mathsf{P/poly}$ or even $\mathsf{EXP}\subset \mathsf{P/poly}$? \ Also, is there a \textit{relativizing} proof that $\mathsf{FBPP}\subset \mathsf{FP/poly}$ would have unlikely consequences?

Does $\mathsf{FBPP}_{\operatorname{log}}=\mathsf{FP}$ under some plausible derandomization assumption? \ Is $\mathsf{FBPP} \subset \mathsf{FBPP}_{\operatorname{negl}}$? \ (Recall that $\mathsf{FBPP}_{\operatorname{log}}$ is the subclass of $\mathsf{FBPP}$ where we require $\operatorname{poly}(n,\log 1/\varepsilon)$ time to achieve error $\varepsilon$, while $\mathsf{FBPP}_{\operatorname{negl}}$ is the variant where we require only that the error probability be negligible.)

Are there examples of problems in $\mathsf{FBPP}$, $\mathsf{FBQP}$, $\mathsf{FP/rpoly}$, or the other relational classes studied in this paper, for which $\operatorname{poly}(1/\varepsilon)$ (rather than, say, $\operatorname{polylog}(1/\varepsilon)$) running time is actually needed to achieve error $\varepsilon$? \ (Certainly we've given examples of \textit{reductions} where such running time is needed.) \ For example, can we show that $\mathsf{FBPP}_{\operatorname{log}} \neq \mathsf{FBPP}$ \textit{unconditionally}, without assuming noncollapse of the polynomial hierarchy?

\section{Acknowledgments}

We thank Daochen Wang and Alexandru Cojocaru for conversations that suggested the $\mathsf{FBQP/poly}$ versus $\mathsf{FBQP/qpoly}$ problem, and Lance Fortnow, Iordanis Kerenidis, Ashley Montanaro, Ronald de Wolf, and David Zuckerman for helpful discussions. \ We thank anonymous referees for constructive feedback.

\bibliographystyle{plain}
\bibliography{thesis}

\begin{thebibliography}{10}

\bibitem{aar:adv}
S.~Aaronson.
\newblock Limitations of quantum advice and one-way communication.
\newblock {\em Theory of Computing}, 1:1--28, 2005.
\newblock Earlier version in CCC'2004. quant-ph/0402095.

\bibitem{aar:pp}
S.~Aaronson.
\newblock Quantum computing, postselection, and probabilistic polynomial-time.
\newblock {\em Proc. Roy. Soc. London}, A461(2063):3473--3482, 2005.
\newblock quant-ph/0412187.

\bibitem{aar:ph}
S.~Aaronson.
\newblock {BQP} and the polynomial hierarchy.
\newblock In {\em Proc. ACM STOC}, 2010.
\newblock arXiv:0910.4698.

\bibitem{aar:samp}
S.~Aaronson.
\newblock The equivalence of sampling and searching.
\newblock In {\em Proc. Computer Science Symposium in Russia (CSR)}, 2011.
\newblock arXiv:1009.5104, ECCC TR10-128.

\bibitem{aa:struc}
S.~Aaronson and A.~Ambainis.
\newblock The need for structure in quantum speedups.
\newblock In {\em Proc. Innovations in Theoretical Computer Science (ITCS)},
  2011.
\newblock arXiv:0911.0996.

\bibitem{aark}
S.~Aaronson and A.~Arkhipov.
\newblock The computational complexity of linear optics.
\newblock {\em Theory of Computing}, 9(4):143--252, 2013.
\newblock Earlier version in STOC'2011. ECCC TR10-170, arXiv:1011.3245.

\bibitem{achen}
S.~Aaronson and L.~Chen.
\newblock Complexity-theoretic foundations of quantum supremacy experiments.
\newblock In {\em Proc. Conference on Computational Complexity}, 2017.
\newblock arXiv:1612.05903.

\bibitem{agunn}
S.~Aaronson and S.~Gunn.
\newblock On the classical hardness of spoofing linear cross-entropy
  benchmarking.
\newblock {\em Theory of Computing}, 16:1--8, 2020.
\newblock arXiv:1910.12085.

\bibitem{ak}
S.~Aaronson and G.~Kuperberg.
\newblock Quantum versus classical proofs and advice.
\newblock {\em Theory of Computing}, 3(7):129--157, 2007.
\newblock Earlier version in CCC'2007. arXiv:quant-ph/0604056.

\bibitem{adleman:bpp}
L.~Adleman.
\newblock Two theorems on random polynomial time.
\newblock In {\em Proc. IEEE FOCS}, pages 75--83, 1978.

\bibitem{agllv}
D.~Aharonov, X.~Gao, Z.~Landau, Y.~Liu, and U.~Vazirani.
\newblock A polynomial-time classical algorithm for noisy random circuit
  sampling.
\newblock arXiv:2211.03999, 2022.

\bibitem{bjk}
Z.~Bar-Yossef, T.~S. Jayram, and I.~Kerenidis.
\newblock Exponential separation of quantum and classical one-way communication
  complexity.
\newblock {\em SIAM J. Comput.}, 38(1):366--384, 2008.
\newblock Earlier version in STOC'2004. ECCC TR04-036.

\bibitem{berlekamp:polyfac}
E.~Berlekamp.
\newblock Factoring polynomials over large finite fields.
\newblock {\em Math. Comp.}, 24:713--735, 1970.

\bibitem{buhrtor}
H.~Buhrman and L.~Torenvliet.
\newblock Randomness is hard.
\newblock {\em SIAM J. Comput.}, 30(5):1485--1501, 2000.
\newblock Earlier version in Proceedings of CCC'1998.

\bibitem{arute}
F.~Arute et~al.
\newblock Quantum supremacy using a programmable superconducting processor.
\newblock {\em Nature}, 574(7779):505--510, 2019.

\bibitem{feffermankimmel}
B.~Fefferman and S.~Kimmel.
\newblock Quantum vs.\ classical proofs and subset verification.
\newblock In {\em Mathematical Foundations of Computer Science}, pages 1--23,
  2018.
\newblock arXiv:1510.06750.

\bibitem{gaoetal}
X.~Gao, M.~Kalinowski, C.~N. Chou, M.~D. Lukin, B.~Barak, and S.~Choi.
\newblock Limitations of linear cross-entropy as a measure for quantum
  advantage.
\newblock arXiv:2112.01657, 2021.

\bibitem{gkkrw}
D.~Gavinsky, J.~Kempe, I.~Kerenidis, R.~Raz, and R.~de Wolf.
\newblock Exponential separation for one-way quantum communication complexity,
  with applications to cryptography.
\newblock {\em SIAM J. Comput.}, 38(5):1695--1708, 2008.
\newblock Earlier version in STOC'2007. quant-ph/0611209.

\bibitem{goldreich:pbpp}
O.~Goldreich.
\newblock In a world of {$P=BPP$}.
\newblock {\em Studies in Complexity and Cryptography}, 6650:191--232, 2011.
\newblock ECCC TR10-135.

\bibitem{gossetsmolin}
D.~Gosset and J.~A. Smolin.
\newblock A compressed classical description of quantum states.
\newblock In {\em Proc. Theory of Quantum Computation, Communication, and
  Cryptography (TQC)}, pages 1--9, 2019.
\newblock arXiv:1801.05721.

\bibitem{holevo}
A.~S. Holevo.
\newblock Some estimates of the information transmitted by quantum
  communication channels.
\newblock {\em Problems of Information Transmission}, 9:177--183, 1973.
\newblock English translation.

\bibitem{ilw23}
R.~Ilango, J.~Li, and R.~Williams.
\newblock Indistinguishability obfuscation, range avoidance, and bounded
  arithmetic.
\newblock In {\em Proc. ACM STOC}, pages 1076--1089, 2023.
\newblock ECCC TR23-038.

\bibitem{kl}
R.~M. Karp and R.~J. Lipton.
\newblock Turing machines that take advice.
\newblock {\em Enseign. Math.}, 28:191--201, 1982.
\newblock Earlier version in STOC'1980, pp. 302-309.

\bibitem{kkd}
N.~Kumar, I.~Kerenidis, and E.~Diamanti.
\newblock Experimental demonstration of quantum advantage for one-way
  communication complexity surpassing best-known classical protocol.
\newblock {\em Nature Communications}, 10(1):1--10, 2019.
\newblock arXiv:1811.09154.

\bibitem{Man19-thesis}
Ryan~L. Mann.
\newblock {\em Quantum computation and combinatorial structures}.
\newblock PhD thesis, University of Technology Sydney, February 2019.

\bibitem{montanaro:lb}
A.~Montanaro.
\newblock Quantum states cannot be transmitted efficiently classically.
\newblock {\em Quantum}, 3(154), 2019.
\newblock arXiv:1612.06546.

\bibitem{my23-adv-owfs}
Tomoyuki Morimae and Takashi Yamakawa.
\newblock Quantum advantage from one-way functions.
\newblock arXiv:2302.04749, 2023.

\bibitem{natarajannirkhe}
A.~Natarajan and C.~Nirkhe.
\newblock A classical oracle separation between {QMA} and {QCMA}.
\newblock arXiv:2210.15380, 2022.

\bibitem{ny}
H.~Nishimura and T.~Yamakami.
\newblock Polynomial time quantum computation with advice.
\newblock {\em Inform. Proc. Lett.}, 90:195--204, 2003.
\newblock ECCC TR03-059, quant-ph/0305100.

\bibitem{gpt4}
OpenAI.
\newblock {GPT-4} technical report.
\newblock arXiv:2303.08774, 2023.

\bibitem{panchenzhang}
F.~Pan, K.~Chen, and P.~Zhang.
\newblock Solving the sampling problem of the {S}ycamore quantum circuits.
\newblock {\em Phys. Rev. Lett.}, 129(9), 2022.
\newblock arXiv:2111.03011.

\bibitem{shamir}
A.~Shamir.
\newblock {IP=PSPACE}.
\newblock {\em J. of the ACM}, 39(4):869--877, 1992.
\newblock Earlier version in FOCS'1990, pp. 11-15.

\bibitem{stockmeyer}
L.~J. Stockmeyer.
\newblock The complexity of approximate counting.
\newblock In {\em Proc. ACM STOC}, pages 118--126, 1983.

\bibitem{yz}
T.~Yamakawa and M.~Zhandry.
\newblock Verifiable quantum advantage without structure.
\newblock In {\em Proc. IEEE FOCS}, pages 69--74, 2022.
\newblock arXiv:2204.02063.

\bibitem{yao:lemma}
A.~C-C. Yao.
\newblock Probabilistic computations: Toward a unified measure of complexity.
\newblock In {\em Proc. IEEE FOCS}, pages 222--227, 1977.

\bibitem{ustc}
H.-S. Zhong, H.~Wang, Y.-H. Deng, M.-C. Chen, L.-C. Peng, Y.-H. Luo, J.~Qin,
  D.~Wu, X.~Ding, Y.~Hu, P.~Hu, X.-Y. Yang, W.-J. Zhang, H.~Li, Y.~Li,
  X.~Jiang, L.~Gan, G.~Yang, L.~You, Z.~Wang, L.~Li, N.-L. Liu, C.-Y. Lu, and
  J.-W. Pan.
\newblock Quantum computational advantage using photons.
\newblock {\em Science}, 370(6523):1460--1463, 2020.
\newblock arXiv:2012.01625.

\end{thebibliography}

\appendix

\section{Conversation with GPT-4\label{GPT}}

Below is the transcript of a conversation, dated January 9, 2023, between one of the authors and a preliminary version of GPT-4 \cite{gpt4} about possible definitions for the class $\mathsf{FBPP}$. \ For ease of reading, the conversation has been \TeX ified (e.g.\ ``FBPP'' $\rightarrow$ $\mathsf{FBPP}$, ``eps'' $\rightarrow \varepsilon$). \ With mild prodding, GPT eventually settles on Definition \ref{FBPPdef}.

\begin{drama}
  \Character{Scott}{scott}
  \Character{GPT}{gpt}

  \scottspeaks: What's a good definition of $\mathsf{FBPP}$ -- the generalization of $\mathsf{BPP}$ to
relational problems?  

  \gptspeaks: $\mathsf{FBPP}$ could be defined as follows:
$\mathsf{FBPP}$ (Functional $\mathsf{BPP}$) is a complexity class for relational problems,
defined analogously to $\mathsf{BPP}$ for decision problems. \ A relational problem
is in $\mathsf{FBPP}$ if there is a probabilistic polynomial-time algorithm that,
given an input, produces an output consistent with the relation with
probability at least $2/3$ (or any other constant bounded away from
$1/2$).

\scottspeaks: But unlike with $\mathsf{BPP}$, amplification isn't possible, so the choice
of constant (such as $2/3$) actually matters, right?

\gptspeaks: The choice of constant does matter, but amplification is still possible in
many cases. \ Just like with $\mathsf{BPP}$, the algorithm could be run multiple
times to increase the confidence in the output. \ However, the
amplification process is more complicated for $\mathsf{FBPP}$, since it's not
always clear how to combine the outputs from multiple runs of the
algorithm to produce a single, more accurate output. \ This will depend
on the specifics of the relational problem at hand. \ But in many cases,
amplification should still be possible.

\scottspeaks: But in other cases not! \ Thus, can you suggest a definition of $\mathsf{FBPP}$
that would allow arbitrary error
probabilities?

\gptspeaks: One possible alternative definition could be: $\mathsf{FBPP}$ (Functional
$\mathsf{BPP}$) is a complexity class for relational
problems. \ A relational problem is in $\mathsf{FBPP}$ if there is a probabilistic
polynomial-time algorithm that, given an input, produces an output
consistent with the relation with probability strictly greater than
$1/2$.

\scottspeaks: Or what about with probability $1-\varepsilon$ for any specified
$\varepsilon > 0$?

\gptspeaks: Yes, that would work as well. \ The definition
could be: $\mathsf{FBPP}$ (Functional $\mathsf{BPP}$) is a complexity class for relational
problems. \ A relational problem is in $\mathsf{FBPP}$ if there is a probabilistic
polynomial-time algorithm that, given an input, produces an output
consistent with the relation with probability at least $1-\varepsilon$ for any
specified $\varepsilon > 0$.

\scottspeaks: And how is the running time of the algorithm allowed
to depend on $\varepsilon$?

\gptspeaks: The definition above doesn't specify, but one reasonable choice
would be to require that
the running time be polynomial in both the input size and $1/\varepsilon$. \ This
would ensure that the algorithm is efficient even as the error
probability is made arbitrarily small.
\end{drama}

\end{document}